\documentclass[11pt]{article}
\usepackage[a4paper,lmargin=1in,rmargin=1in]{geometry}
\usepackage{amsmath}
\usepackage{amsfonts}
\usepackage{amsbsy}
\usepackage{amsthm}
\usepackage{amssymb}
\usepackage{accents}
\usepackage{tensor}
\usepackage{enumitem}
\newtheorem{thrm}{Theorem}

\newtheorem{lmm}[thrm]{Lemma}
\newtheorem{dfn}[thrm]{Definition}
\theoremstyle{definition}
\newtheorem{rmrk}{Remark}

\DeclareMathOperator{\End}{End}
\DeclareMathOperator{\id}{id}
\DeclareMathOperator{\tr}{tr}
\DeclareMathOperator{\Tr}{Tr}
\DeclareMathAlphabet{\mathpzc}{OT1}{pzc}{m}{it}

\newcommand{\rmi}{\mathrm{i}}
\newcommand{\rmd}{\mathrm{d}}

\newcommand{\symped}[1]{\accentset{S}{#1}}
\newcommand{\sympman}{\mathcal{M}}
\newcommand{\bund}{\mathcal{E}}
\newcommand{\ebund}{{\End(\bund)}}
\newcommand{\wbund}{W}
\newcommand{\dotmark}{\mathcal{X}}
\newcommand{\wbof}[1]{{\wbund_{#1}}}
\newcommand{\wbc}{\wbof{\mathbb{C}}}
\newcommand{\wbb}{\wbof{\bund}}
\newcommand{\wbe}{\wbof{\ebund}}
\newcommand{\wbdot}{\wbof{\dotmark}}
\newcommand{\wdof}[1]{{\wbund^D_{#1}}}

\newcommand{\wddot}{\wdof{\dotmark}}
\newcommand{\muof}[1]{{\mu^{#1}}}
\newcommand{\muc}{\muof{\mathbb{C}}}
\newcommand{\mue}{\muof{\ebund}}
\newcommand{\mudot}{\muof{\dotmark}}
\newcommand{\qmof}[1]{{Q_{#1}}}
\newcommand{\qmc}{\qmof{\mathbb{C}}}
\newcommand{\qme}{\qmof{\ebund}}
\newcommand{\qmb}{\qmof{\bund}}
\newcommand{\qmdot}{\qmof{\dotmark}}
\newcommand{\qminvof}[1]{{Q^{-1}_{#1}}}
\newcommand{\qminvc}{\qminvof{\mathbb{C}}}
\newcommand{\qminve}{\qminvof{\ebund}}
\newcommand{\qminvb}{\qminvof{\bund}}
\newcommand{\qminvdot}{\qminvof{\dotmark}}
\newcommand{\rsof}[1]{{r^{#1}}}
\newcommand{\rsc}{\rsof{\mathbb{C}}}
\newcommand{\rse}{\rsof{\ebund}}
\newcommand{\rsb}{\rsof{\bund}}
\newcommand{\rsdot}{\rsof{\dotmark}}
\newcommand{\gamdot}{{\gamma_{\dotmark}}}
\newcommand{\hdot}{H_\dotmark}
\newcommand{\gtder}{\Delta_t}
\newcommand{\brof}[1]{{R^{#1}}}
\newcommand{\brc}{\brof{\mathbb{C}}}
\newcommand{\bre}{\brof{\ebund}}
\newcommand{\brdot}{\brof{\dotmark}}
\newcommand{\gbr}{\gend{R}}
\newcommand{\gbrof}[1]{{\gbr^{#1}}}
\newcommand{\gbrc}{\gbrof{\mathbb{C}}}
\newcommand{\gbre}{\gbrof{\ebund}}
\newcommand{\gbrdot}{\gbrof{\dotmark}}
\newcommand{\connsymp}{\partial^{S}}
\newcommand{\connbund}{\partial^{\bund}}
\newcommand{\connend}{\partial^{\ebund}}
\newcommand{\gambund}{\Gamma^{\bund}}

\newcommand{\curvbund}{R^{\bund}}
\newcommand{\curvsymp}{\symped{R}}

\newcommand{\gcirc}{\stackrel{\sim}{\circ}}
\newcommand{\gend}[1]{\accentset{\sim}{#1}}
\newcommand{\gconn}{\gend{\partial}}
\newcommand{\ccomm}{\stackrel{\circ}{,}}
\newcommand{\gcomm}{\stackrel{\gcirc}{,}}
\newcommand{\gcommt}[1]{\stackrel{\gcirc_{#1}}{,}}
\newcommand{\ibh}{\frac{\rmi}{h}}
\newcommand{\adjact}{\mathfrak{T(g)}}
\newcommand{\adjacthmtp}{\mathfrak{T(g(}t))}

\begin{document}

\title{\textsc{Remarks on generalized Fedosov algebras}}
\date{}
\author{Micha{\l} Dobrski\footnote{michal.dobrski@p.lodz.pl}
\\
\small
\emph{Institute of Physics}
\\
\small
\emph{Lodz University of Technology}
\\
\small
\emph{W\'olcza\'nska 219, 90-924 {\L}\'od\'z, Poland}}
\maketitle
\abstract{The variant of Fedosov construction based on fairly general fiberwise product in the Weyl bundle is studied. We analyze generalized star products of functions, of sections of endomorphisms bundle, and those generating deformed bimodule structure as introduced previously by Waldmann. Isomorphisms of generalized Fedosov algebras are considered and their relevance for deriving Seiberg-Witten map is described. The existence of the trace functional is established. For star products and for the trace functional explicit expressions, up to second power of deformation parameter, are given. The example of symmetric part of noncommutativity tensor is discussed as a case with possible field-theoretic application.}
\section{Introduction}
Fedosov quantization \cite{fedosovart}, \cite{fedosovbook} is the beautiful construction providing powerful tools for studying various aspects of deformation quantization on symplectic manifolds. Quite recently, in author's previous works \cite{dobrski-ncgr},\cite{dobrski-ncgrreal}, it was advocated that Fedosov formalism is well suited for building global, geometric field theories on noncommutative spacetimes. The reason for such claim is that Fedosov theory admits generalization producing geometric star products in endomorphism bundles of vector bundles, and this is the place where (some) gauge fields can be put to live. Moreover, there exist convenient isomorphism theory, as well as trace functional construction for Fedosov quantization. These structures turn out to encode local and global versions of Seiberg-Witten map \cite{dobrski-sw}, \cite{dobrski-ncgr}, which can be viewed as a quite natural consequence of Fedosov formalism (without postulating it separately). On the other hand, generic Fedosov construction seems to be in some aspects still too rigid for field-theoretic applications. This is because it is ``canonical'' -- taking minimal input of geometric data (symplectic form and symplectic connection) it produces simplest geometric deformation quantization. But such canonicality means that the construction ``knows nothing'' about other structures on underlying manifold, eg. about metric. 
The aim of the present paper is to introduce into Fedosov formalism fairly general additional degrees of freedom, which can be interpreted in various ways, possibly also in field-theoretic context. This is achieved by modifying core structure of Fedosov construction, namely the star product $\circ$ in the fibers of tangent bundle. In our setting it could be multiplication different from the Moyal one. There were various approaches investigating some specific non-Moyal fiberwise products, eg. \cite{tosiekacta3}, \cite{bordwald}, \cite{dolglyaksharap}, \cite{bering}. Here, we especially follow that of \cite{tosiekacta3}, but in more general context. The construction is carried out to further point then usually -- the isomorphism theory and the trace functional are studied, as well as generalizations to products involving sections of vector and endomorphism bundles, as in \cite{waldmann}. Being primarily interested in further applications, we provide explicit formulas for generalized star-products and trace functional, up to second power of deformation parameter.    

It should be stressed that the present paper does not aim at providing most general variant of Fedosov quantization. The theory developed by Fedosov is very well understood now and important generalizations in directions not covered here were analyzed. Notable example is related to quantization of irregular Poisson manifolds. In \cite{cattfldr} elements of Fedosov's technology were used for explicit globalization of Kontsevich star-product. On the other hand, in \cite{lyakhbrst} the Fedosov construction was identified with BRST quantization of certain constrained system and this observation, together with earlier ideas from \cite{lyakhpi}, were later used in \cite{lyakhqs} for quantization of large class of irregular Poisson manifolds. Here, we are not reaching beyond original symplectic setting and concentrate only on exploring consequences of modifying fiberwise $\circ$ product.

The paper is organized as follows. First (section 2), basic structure of generalized Fedosov construction is described. The existence of star-products for factors of various type (function, endomorphism, vector) is established. Then (section 3), the isomorphism theory is analyzed in some detail, as it is indispensable component for deriving Seiberg-Witten relations (discussed in the same section) and for constructing trace functional. The latter one is studied in section 4. As has been mentioned before, the special attention is paid to explicit formulas (including one for the trace functional), which are presented in section 5. It also contains example which illustrates single, very specific interpretation of the present generalization -- introduction of symmetric part of noncommutativity tensor. Some concluding remarks are given in section 6.

\section{Generalized Fedosov construction}
In this section the variant of Fedosov construction is described. The reader interested in studying conceptual structure of Fedosov quantization (both geometric and algebraic) is referred to \cite{emmrwein}, \cite{xu}, \cite{farkas}, \cite{tosiek}. Here, taking original formulation of \cite{fedosovbook} together with generalization developed in \cite{waldmann}, we extend them using methods of \cite{tosiekacta3}. The main concept of this extension can be summarized as follows -- replace fiberwise Moyal product $\circ$ in the Weyl bundle by some other product $\gcirc$, and allow it to vary across the fibers. As all fiberwise products must be equivalent to Moyal product, the required $\gcirc$ can be introduced by choosing an isomorphism $g$ to the standard Moyal algebra. This technique was used in \cite{tosiekacta3} for the description of deformation quantizations originating in different operator orderings. We make use of it to introduce fairly general fiberwise product $\gcirc$. 
\subsection{Weyl bundles and their basic properties}
The initial data for the original Fedosov construction are given by the following structures. Let $(\sympman, \omega, \connsymp)$ be the Fedosov manifold \cite{gelfand,bielgutt} of dimension $2n$, for which components of symplectic curvature tensor will be denoted by $\tensor{\curvsymp}{^i_{jkl}}$. Consider a finite dimensional complex vector bundle $\bund$ over $\sympman$. Let $\connbund$  be a linear connection in $\bund$ and $\connend$ connection induced by $\connbund$ in endomorphisms bundle $\ebund$. We are going to use $\curvbund_{kl}$ to denote components of curvature of $\connbund$.  They are given by the local formula $\curvbund_{kl}=\frac{\partial \gambund_l}{\partial x_k}-\frac{\partial \gambund_k}{\partial x_l}+[\gambund_k,\gambund_l]$ for $\connbund=d+\gambund$ with $\gambund=\gambund_i \rmd x^i$ being local connection $1$-form. As it will be seen, for the purpose of our generalization some extra fields will be required.

Introduce over $\sympman$ the Weyl bundle $\wbc$. Also, let us have $\wbb=\wbc \otimes \bund$ and $\wbe=\wbc \otimes \ebund$. Sections of $\wbc \otimes \Lambda$, $\wbb \otimes \Lambda$ and $\wbe \otimes \Lambda$ can be locally written as a formal sums
\begin{equation}
\label{wsect}
a(x,y,h)=\sum_{k,p,q \geq 0} h^k a_{k,i_1 \dots i_p j_1 \dots j_q}(x)y^{i_1} \dots y^{i_p} \rmd x^{j_1} \wedge \dots \wedge \rmd x^{j_q}
\end{equation}
where $a_{k,i_1 \dots i_p j_1 \dots j_q}(x)$ are components of some (respectively) $\mathbb{C}$-, $\bund$- or $\ebund$-valued covariant tensor field at $x \in \sympman$, and $y \in T_x\sympman$. We will repeatedly encounter statements that hold true for all variants of these target spaces. To avoid redundant repetitions let us introduce notation that $\dotmark$ stands for $\mathbb{C}$, $\bund$ and $\ebund$ if not otherwise restricted. 

The ``degree counting'' combined with ``iteration method'' are basic tools of Fedosov construction for controlling behavior of formal series (\ref{wsect}). Consider monomial in (\ref{wsect}) with $p$-fold $y^i$ and $k$-th power of $h$. One can prescribe degree to it by the rule
\begin{equation}
\label{degdef}
\deg (h^k a_{k,i_1 \dots i_p j_1 \dots j_q}(x)y^{i_1} \dots y^{i_p} \rmd x^{j_1} \wedge \dots \wedge \rmd x^{j_q})=2k+p
\end{equation}
For a general inhomogeneous element of Weyl bundle, the degree is defined as the lowest degree of its nonzero monomials. The iteration method can be described in the following way. Let $P_m$ denote the operator which extracts monomials of degree $m$ from given $a$
\begin{equation*}
P_m(a)(x,y,h)=\sum_{2k+p=m} h^k a_{k,i_1 \dots i_p j_1 \dots j_q}(x)y^{i_1} \dots y^{i_p} \rmd x^{j_1} \wedge \dots \wedge \rmd x^{j_q}.
\end{equation*}
We frequently consider equations of the form 
\begin{equation}
\label{fedoitermet}
a=b +K(a)
\end{equation} 
and try to solve them iteratively with respect to $a$, by putting $a^{(0)}=b$ and $a^{(n)}=b+K(a^{(n-1)})$. If $K$ is linear and raises degree (i.e.\ $\deg a < \deg K(a)$ or $K(a)=0$) then one can quickly deduce that the unique solution of (\ref{fedoitermet}) is given by the series of relations $P_m(a)=P_m(a^{(m)})$. However, in the case of nonlinear $K$ (as in (\ref{fedabelrecur})) the more careful analysis must be performed.

The useful property of Weyl bundles is the existence of global ``Poincare decompostion''. It can be verified that operators $\delta$ and $\delta^{-1}$ defined by relations 
\begin{equation}
\delta a =\rmd x^k \wedge \frac{\partial a}{\partial y^k} \quad \text{and} \quad \delta^{-1}a_{km}=\frac{1}{k+m}y^s \iota \left(\frac{\partial}{\partial x^s}\right) a_{km}
\end{equation}
for $a_{km}$ with $k$-fold $y$ and $m$-fold $\rmd x$ (and by linear extension for inhomogeneous $a$) provide for arbitrary $a \in \wbdot \otimes \Lambda$ decompostion
\begin{equation}
\label{pncrdecomp}
a=\delta \delta^{-1}a + \delta^{-1} \delta a + a_{00}
\end{equation}
where $a_{00}$ denotes homogeneous part of $a$ containing no $y^i$ and $\rmd x^j$. Both $\delta$ and $\delta^{-1}$ are nilpotent, ie{.} $\delta\delta=\delta^{-1} \delta^{-1}=0$.

Let $\circ$ denote usual fiberwise Moyal product 
\begin{equation}
\label{moyal}
a \circ b = \sum_{m=0}^{\infty}\left( -\frac{\rmi h}{2}\right)^m \frac{1}{m!} 
\frac{\partial^m a}{\partial y^{i_1} \dots \partial y^{i_m}}
\omega^{i_1 j_1} \dots \omega^{i_m j_m}
\frac{\partial^m b}{\partial y^{j_1} \dots \partial y^{j_m}}
\end{equation}
Here $\omega^{ij}$ are components of the Poisson tensor corresponding to symplectic form $\omega=\frac{1}{2}\omega_{ij} \rmd x^i \wedge \rmd x^j$. 
Notice that above formula is meaningful not only for $a,b \in \wbc \otimes \Lambda$ or $a,b \in \wbe \otimes \Lambda$. Following \cite{waldmann} we admit case of $a \in \wbe \otimes \Lambda$, $b \in \wbb \otimes \Lambda$ which would provide deformation of action of an endomorphism on a vector, and $a \in \wbe\otimes\Lambda$, $b \in \wbc\otimes\Lambda$ corresponding to deformation of scaling a vector field by a function\footnote{We will also allow case of $a \in \wbc\otimes\Lambda$, $b \in \wbe\otimes\Lambda$ but this fiberwise product does not correspond to deformation of right multiplication of a vector field by a function on the manifold.}. For all these cases the Moyal product is associative $(a \circ b) \circ c=a \circ (b \circ c)$, as long as both sides of this relation are well defined. 
The generalized deformed fiberwise product can be introduced in the following way. Let 
\begin{equation}
\label{gdef}
g=\id+\sum_{\substack{2s-k \geq 0 \\ s,k > 0}} h^s g_{(s)}^{i_1 \dots i_k} \frac{\partial^k}{\partial y^{i_1}\dots \partial y^{i_k}}
\end{equation}
be fiberwise star equivalence isomorphism  with $g_{(s)}^{i_1 \dots i_k}$ being components of some $k$-contravariant $\mathbb{C}$-valued tensors on $\sympman$. Formula (\ref{gdef}) implies that $g$ preserves degree of elements of Weyl bundle\footnote{Consider arbitrary $a \neq 0$. Let $\pi(a)$ be monomial in $a$ of lowest degree (ie. $\deg \pi(a) =\deg a$), and with maximal number of $y$'s within this degree. Then it can be easily verified that $\pi(a)=\pi(ga)$ (because of the form of derivatives in (\ref{gdef})), hence $\deg (ga) = \deg (a)$.}, ie{.} $\deg(ga) = \deg(a)$, and that $g$ is formally invertible. Indeed, writing 
\begin{equation*}
g^{-1}=\id+\sum_{\substack{2s-k \geq 0 \\ s,k > 0}} h^s \breve{g}_{(s)}^{i_1 \dots i_k} \frac{\partial^k}{\partial y^{i_1}\dots \partial y^{i_k}}
\end{equation*}
and considering relation $ g^{-1} g =\id$ one arrives at the formula
\begin{equation}
\label{ginvcoefs}
g_{(a)}^{i_1 \dots i_b} + \breve{g}_{(a)}^{i_1 \dots i_b}+ \sum_{\substack{2a-b \geq 2s-k \geq 0 \\ a > s > 0\\b>k>0}} g_{(s)}^{(i_1 \dots i_k} \breve{g}_{(a-s)}^{i_{k+1} \dots i_b)} = 0
\end{equation}
which allows recursive computation of coefficients $\breve{g}_{(a)}^{i_1 \dots i_b}$. The generalized fiberwise product is defined as 
\begin{equation}
a \gcirc b=g^{-1}(ga \circ gb)
\end{equation}
We are going to do denote graded commutator with respect to $\gcirc$ by
\begin{equation}
[a \gcomm b]=a \gcirc b -(-1)^{rs} b \gcirc a
\end{equation}
for $r$-form $a$ and $s$-form $b$. Similarly $[a \ccomm b]$ stands for commutator with respect to $\circ$.
Form of (\ref{gdef}) implies that $g$ commutes with $\delta$ (ie{.} $g \delta=\delta g$) and consequently $\delta$ is  the $+1$-derivation with respect to $\gcirc$
\begin{equation}
\delta (a \gcirc b)=\delta a \gcirc b + (-1)^k a \gcirc \delta b
\end{equation}
for $k$-form $b$. The operator $\delta$ can be represented as a commutator with respect to $\circ$ by $\delta a= \ibh[s \ccomm a]$ with $s=-\omega_{ij}y^i \rmd x^j$. It follows that it is also commutator with respect to $\gcirc$ because
\begin{equation*}
\delta a =g^{-1} \delta g a = \ibh g^{-1}[s\ccomm ga]=\ibh[\gend{s} \gcomm a]
\end{equation*}
where $\gend{s}=g^{-1} s$.

Let us notice that in the case of $\wbe$ one must deal with following subtlety related to initial noncommutativity of product of endomorphisms. For operators of the form $K=\ibh [s \gcomm \cdot \,]$ one cannot use arbitrary $s$, as this could yield negative powers of $h$. The only appropriate $s \in C^\infty (\sympman, \wbe \otimes \Lambda)$ are those for which monomials of the form $h^0 a_{0,i_1 \dots i_p j_1 \dots j_q} y^{i_1} \dots y^{i_p} \rmd x^{j_1} \wedge \dots \wedge \rmd x^{j_q}$ are defined by central endomorphisms $a_{0,i_1 \dots i_p j_1 \dots j_q}$. (This statement can be easily verified for Moyal product $\circ$, and then transported to generalized case by means of $g$). Let us call them $C$-sections. Also, let us use the term  $C$-operator for mappings which transport  $C$-sections to $C$-sections. Obviously $\delta$, $\delta^{-1}$, 
$g$ and $g^{-1}$ are $C$-operators. The following lemma is an useful tool for controlling occurrence of negative powers of $h$. 
\begin{lmm}
\label{cs_lmm}
For arbitrary $C$-section $s$ the commutator $\ibh[s \gcomm \cdot \,]$ is a $C$-operator.
\end{lmm}
The proof is straightforward for $\circ$, and using $C$-operator $g$ one can immediately extend it to generalized product $\gcirc$. 
 
We should also mention that the only central elements of $\wbe \otimes \Lambda$ and $\wbc \otimes \Lambda$ are these belonging to $\Lambda$, ie{.} scalar forms on base manifold $\sympman$. This fact is well known for $\circ$ and can be trivially transfered to the case of $\gcirc$.
 
Connections $\connsymp$ and $\connbund$ give rise to the connections $\partial^\wbdot$ in all variants of Weyl bundle. Let $\tensor{\Gamma}{^i_{jk}}$ be local connection coefficients of $\connsymp$, and let $\gambund$ be local connection $1$-form of $\connbund$, ie{.} locally $\connbund=\rmd+\gambund$ and $\connend=\rmd+[\gambund,\cdot\,]$. Then locally, in some Darboux coordinates, one can write
\begin{subequations}
\label{connlocalform}
\begin{align}
\label{connlocalformwbc}
\partial^\wbc& =\rmd + \ibh [1/2 \Gamma_{ijk}y^i y^j \rmd x^k\ccomm \cdot\,] \\
\label{connlocalformwbb}
\partial^\wbb& =\rmd + \ibh [1/2 \Gamma_{ijk}y^i y^j \rmd x^k\ccomm \cdot\,] + \gambund \\
\label{connlocalformwbe}
\partial^\wbe& =\rmd + \ibh[1/2 \Gamma_{ijk}y^i y^j \rmd x^k-\rmi h \gambund \ccomm \cdot\,]
\end{align}
\end{subequations}
using local connection coefficients $\Gamma_{ijk}=\omega_{is}\tensor{\Gamma}{^s_{jk}}$ of $\connsymp$.
Notice that if we consider $\wbc$ as a subbundle of $\wbe$, then (\ref{connlocalformwbe}) gives same results as (\ref{connlocalformwbc}). Connections $\partial^\wbdot$ are $+1$-derivations for Moyal product, and since we admitted various types of factors for $\circ$, this statement includes all corresponding ``compatibilities'' (eg{.} $\partial^\wbb (a \circ b)=\partial^\wbe a \circ b + (-1)^k a \circ \partial^\wbb b$ for $a \in \wbe \otimes \Lambda^k$ and $b \in \wbb \otimes \Lambda$). We need analogous $+1$-derivations for $\gcirc$, thus let us introduce the generalized connections
\begin{equation}
\label{gconns}
\gconn^\wbdot= g^{-1} \partial^\wbdot g
\end{equation}
for which the relation $\gconn^\wbdot (a \gcirc b)=\gconn^\wbdot a \gcirc b +(-1)^k a \gcirc \gconn^\wbdot b$ holds. In general $\gconn^\wbdot$ cannot be written in form analogous to (\ref{connlocalform}) due to derivatives of fields $g_{(s)}^{i_1 \dots i_k}$. However, for $\wbc$ and $\wbe$ from $\big(\partial^\wbdot\big)^2 =\ibh[\brdot \ccomm \cdot\,]$ with $\brc=\frac{1}{4}\omega_{im}\tensor{\curvsymp}{^m_{jkl}}y^i y^j \rmd x^k \wedge \rmd x^l$ and $\bre=\brc-\frac{ih}{2} \curvbund_{kl} \rmd x^k \wedge \rmd x^l$ we have
\begin{equation}
\label{gconnsqr}
\big(\gconn^\wbdot\big)^2 = g^{-1} \partial^\wbdot \partial^\wbdot g = \ibh g^{-1} [\brdot \ccomm g(\cdot)]=\ibh[\gbrdot \gcomm \cdot\,]
\end{equation}
for $\gbrdot=g^{-1}\brdot$. Notice that by lemma \ref{cs_lmm} both $\partial^\wbe$ and $\gconn^\wbe$ are $C$-operators, while $\bre$ and $\gbre$ are $C$-sections.

\subsection{First Fedosov theorem -- Abelian connections}
One can analyze more general connections in $\wbc$ and $\wbe$ of the form
\begin{equation}
\label{wbconntype}
\nabla^\wbdot=\gconn^\wbdot + \ibh[ \gamdot \gcomm \cdot \,]
\end{equation}
with $\gamdot \in C^\infty (\sympman, \wbdot \otimes \Lambda^1)$. It follows that $\big(\nabla^\wbdot\big)^2 = \ibh [ \Omega^\wbdot \gcomm \cdot \,]$ for curvature $2$-form defined as
\begin{equation}
\label{weylcurvdef}
\Omega^\wbdot = \gbrdot + \gconn^\wbdot \gamdot + \ibh \gamdot \gcirc \gamdot
\end{equation}
Abelian (flat) connections in Weyl bundles will be denoted by $D^\wbdot$. Flatness conditions $\big(D^\wbc\big)^2=0$,  $\big(D^\wbe\big)^2=0$ imply that for Abelian connections $D^\wbc$ and $D^\wbe$ their curvatures must be  scalar $2$-forms. The first essential element of Fedosov formalism is explicit (although recursive) construction of Abelian connections. In our context this result can be summarized in the following theorem.

\begin{thrm}
\label{fedo_abel}
Let $\gconn^\wbdot$ be arbitrary connections of type (\ref{gconns}), let $\kappa \in C^\infty (\sympman, \Lambda^2)[[h]]$ be formal power series of closed $2$-forms (ie{.} $\rmd \kappa=0$), and let 
\begin{equation*}
\muc \in C^\infty (\sympman, \wbc), \quad \quad \mue=\muc -\rmi h \Delta\mue \in C^\infty (\sympman, \wbe)
\end{equation*}
 be arbitrary $C$-sections such that $\deg \mudot \geq 3$ and $\mudot|_{y=0}=0$. There exist unique Abelian connections
\begin{subequations}
\label{abelecdef}
\begin{align}
D^\wbc& =-\delta + \gconn^\wbc + \ibh[\rsc \gcomm \cdot\,]\\
D^\wbe& =-\delta + \gconn^\wbe + \ibh[\rse \gcomm \cdot\,]
\end{align}
\end{subequations}
satisfying
\begin{itemize}
\item $\Omega^\wbc = \Omega^\wbe=-\omega + h\kappa$
\item $\rse$ is $C$-section,
\item $\delta^{-1}\rsdot=\mudot$, 
\item $\deg \rsdot \geq 2$.
\end{itemize}
The $1$-forms $\rsc \in C^\infty (\sympman, \wbc\otimes\Lambda^1)$, $\rse \in C^\infty (\sympman, \wbe\otimes\Lambda^1)$ can be calculated as the unique solutions of equations
\begin{equation}
\label{fedabelrecur}
\rsdot=r_0^\dotmark + \delta^{-1}(\gconn^\wbdot \rsdot +\ibh \rsdot \gcirc \rsdot)
\end{equation}
with $r_0^\dotmark=\delta^{-1}(\gbrdot-h\kappa)+\delta \mudot$. Moreover $D^\wbc$ and $D^\wbe$ define Abelian connection
\begin{equation}
D^\wbb=-\delta + \gconn^\wbb + \ibh[\rsc \gcomm \cdot\,] + \ibh\rsb
\end{equation}
with $C^\infty (\sympman, \wbe\otimes\Lambda^1) \ni \rsb=\rse-\rsc$ and $\ibh\rsb$ not containing negative powers of $h$, such that following compatibility conditions hold true
\begin{subequations}
\label{fedocompatib}
\begin{align}
D^\wbb(A \gcirc X)& = (D^\wbe A)\gcirc X + (-1)^k A \gcirc D^\wbb X\\
D^\wbb(X \gcirc a)& = (D^\wbb X)\gcirc a + (-1)^l X \gcirc D^\wbc a
\end{align}
\end{subequations}
for $A \in C^\infty (\sympman, \wbe\otimes\Lambda^k)$, $X \in C^\infty (\sympman, \wbb\otimes\Lambda^l)$ and $a \in C^\infty (\sympman, \wbc\otimes\Lambda)$.
\end{thrm}
\begin{proof}[Outline of the proof]
The theorem combines results of Fedosov (theorems 5{.}2{.}2, 5{.}3{.}3 of \cite{fedosovbook}) and Waldmann (theorem 3 of \cite{waldmann}) with present generalized setting prototyped in \cite{tosiekacta3}. There is nothing substantially new in the proof, thus let us restrict to its key ingredients and few points for which some care due to our generalization should be taken.

Notice first that connections (\ref{abelecdef}) can be written as
\begin{equation}
D^\wbdot= \gconn^\wbdot + \ibh[-\gend{s}+\rsdot \gcomm \cdot\,]\\
\end{equation}
and thus are of type (\ref{wbconntype}). Their curvatures (\ref{weylcurvdef}) are given by
\begin{equation}
\label{omeq1}
\Omega^\wbdot=   \gbrdot - \omega - \delta \rsdot + \gconn^\wbdot \rsdot +\ibh \rsdot \gcirc \rsdot
\end{equation}
and the requirement $\Omega^\wbdot=-\omega+h\kappa$ yields
\begin{equation}
\label{omeq2}
\delta \rsdot=   \gbrdot - h\kappa + \gconn^\wbdot \rsdot +\ibh \rsdot \gcirc \rsdot
\end{equation}
This formula together with decomposition (\ref{pncrdecomp}) and condition $\delta^{-1}\rsdot=\mudot$ gives relation (\ref{fedabelrecur}). In order to show that unique recursive solution of (\ref{fedabelrecur}) exists, one could use lemma 5{.}2{.}3 of \cite{fedosovbook}. Its proof relies on ``degree counting'' and stays valid in generalized case. This is because both $g$ and $g^{-1}$ preserve degree.  In turn $\gconn$ does not lower degree and $\deg(a \gcirc b) = \deg(a) + \deg(b)$ if $a,b \neq 0$. Consequently, whole reasoning of the proof remains intact. To ensure that $\rse$ is a $C$-section one may observe that it is calculated by recursive application of $\id + \delta^{-1}\left( \gconn  + \ibh (\cdot)^{2}  \right)$ which for $1$-forms is an $C$-operator. The initial point for this iterative procedure is given by the $C$-section $r_0^\ebund = \delta^{-1}(\gbre-h\kappa)+\delta \mue$. One can check (compare with \cite{fedosovbook}) that $\rsdot$ computed from (\ref{fedabelrecur}) indeed yields curvature $\Omega^\wbdot=-\omega + h \kappa$. Conditions $\rmd \kappa=0$ and $\mudot|_{y=0}=0$ must be used for this purpose. The property $\deg \rsdot \geq 2$ follows from $\deg r_0^\dotmark \geq 2$ and the fact that $\delta^{-1}\left( \gconn  + \ibh (\cdot)^{2}  \right)$ raises degree.

The property $(D^\wbb)^2=0$ can be verified by direct calculation (condition $\Omega^\wbc=\Omega^\wbe$ appears to be important here), as well as ``compatibilities'' (\ref{fedocompatib}). 

The difference between recursions (\ref{fedabelrecur}) for $\rse$ and for $\rsc$ yields following relation
\begin{equation*}
\ibh \rsb=\ibh\left(r_0^\ebund-r_0^\mathbb{C}\right)+\delta^{-1}\left(\gconn^\ebund\left(\ibh \rsb\right)+\left(\ibh \rsb\right)\gcirc\left(\ibh \rsb\right)+\ibh[\rsc\gcomm\ibh \rsb]\right)
\end{equation*}
which is iterative formula for $\ibh \rsb$. It starts with 
\begin{equation*}
\ibh r_0^\bund = \ibh \left( r_0^\ebund-r_0^\mathbb{C}\right)=\frac{1}{2} \delta^{-1}\curvbund_{ij} \rmd x^i \wedge \rmd x^j + \delta \Delta \mue \in \wbe \otimes \Lambda^1
\end{equation*}
and proceeds by application of operator $ \delta^{-1}(\gconn^\ebund + (\cdot)^2 + \ibh[\rsc \gcomm \cdot\,])$ which raises degree and does not produce negative powers of $h$ because $\rsc$ is a $C$-section. We infer that $\ibh \rsb \in \wbe \otimes \Lambda^1$. 
\end{proof}
The immediate consequence of above theorem and lemma \ref{cs_lmm} is that $D^\wbe$ is a $C$-operator.
\subsection{Second Fedosov theorem -- star products}
Abelian connections allow to construct nontrivial liftings of functions, sections of $\bund$ and sections of $\ebund$ to sections of corresponding Weyl bundles. The key point is that these liftings form subalgebra of algebra of all sections. In this way desired star products are constructed -- we take two objects we want to multiply, lift them to Weyl bundles, multiply liftings using $\gcirc$ and project the result back to the appropriate space of functions or sections of suitable bundle.  

More precisely, let us call $a\in C^\infty (\sympman, \wbdot)$ flat if $D^\wbdot a=0$. Flat sections form subalgebra of $C^\infty (\sympman, \wbdot)$ (this is obvious consequence of Leibniz rule and $(D^\wbdot)^2=0$) denoted by $\wddot$. Let $\qmdot(a)$ be the solution of equation
\begin{equation}
b=a+\delta^{-1} (D^\wbdot + \delta)b
\end{equation}
with respect to $b$. The iteration method ensures that $\qmdot : C^\infty (\sympman, \wbdot) \to C^\infty (\sympman, \wbdot)$ is well-defined linear bijection. It follows that the inverse mapping is given by $\qminvdot a = a-\delta^{-1} (D^\wbdot + \delta)a$. The following theorem holds.
\begin{thrm}
$\qmdot$ bijectively maps $C^\infty (\sympman, \dotmark)[[h]]$ to $\wddot$.
\end{thrm}
This is just theorem 5{.}2{.}4 of \cite{fedosovbook}, combined with extensions of \cite{waldmann} and phrased in our generalized context. We omit the proof, as it does not require any changes comparing to original formulation.

Now let us define all variants of generalized Fedosov product. The second Fedosov theorem, properties of $\gcirc$ and $D^\wbdot$ (notice importance of compatibility conditions (\ref{fedocompatib}) for (\ref{spec_star_1}) and (\ref{spec_star_2})) yield that
\begin{subequations}
\label{starproducts}
\begin{align}
f*g &=\qminvc(\qmc f \gcirc \qmc g) && \text{for } f,g \in C^\infty (\sympman, \mathbb{C})[[h]]\\
A*B &=\qminve(\qme A \gcirc \qme B) && \text{for } A,B \in C^\infty (\sympman, \ebund)[[h]]\\
\label{spec_star_1}
A*X &=\qminvb(\qme A \gcirc \qmb X) && \text{for } A \in C^\infty (\sympman, \ebund)[[h]]\text{, }X \in C^\infty (\sympman, \bund)[[h]]\\
\label{spec_star_2}
X*f &=\qminvb(\qmb X \gcirc \qmc f) && \text{for } X \in C^\infty (\sympman, \bund)[[h]]\text{, } f \in C^\infty (\sympman, \mathbb{C})[[h]]
\end{align}
\end{subequations}
are associative (in all meaningful ways) star products. 

Writing decomposition (\ref{pncrdecomp}) for $(D^\wbdot+\delta)a$ and using $(D^\wbdot)^2=\delta^2=0$ one can derive identity
\begin{equation*}
\qminvdot D^\wbdot + \delta \qminvdot=0
\end{equation*}
It follows that corollary 5{.}2{.6} of \cite{fedosovbook} holds true in our context.
\begin{thrm}
\label{triv_deqlemma}
For given $b \in C^\infty (\sympman, \wbdot \otimes \Lambda^p)$, $p>0$  equation $D^\wbdot a=b$ has a solution if and only if $D^\wbdot b=0$. The solution may be chosen in the form $a=-\qmdot \delta^{-1}b$.
\end{thrm}

\section{Isomorphisms of generalized Fedosov algebras}
Now it is time to analyze isomorphisms of generalized Fedosov algebras. Again, the original presentation of \cite{fedosovbook} is followed, supplemented with some necessary additions and modifications. The material is presented in rather detailed way. This is because we want to provide solid basis for further statements on Seiberg-Witten map an the trace functional.
Here, bundles $\wbe$ and $\wbc$ are dealt with, thus we restrict $\dotmark$ to denote $\mathbb{C}$ or $\ebund$ within this section. 

Let $g_t$ be homotopy of isomorphisms of type (\ref{gdef}) parametrized by $t \in [0,1]$. Corresponding homotopy of deformed fiberwise products will be denoted by $\gcirc_t$.  We would like to study evolution with respect to parameter $t$, and for this purpose let us introduce
\begin{equation}
\gtder=g^{-1}_t \frac{d}{dt} g_t= \frac{d}{dt}+g^{-1}_t\frac{d g_t}{dt}
\end{equation}
It can be immediately observed that $\gtder$ is $0$-derivation with respect to $\gcirc_t$
\begin{equation}
\gtder(a \gcirc_t b)=\gtder a \gcirc_t b + a \gcirc_t \gtder b
\end{equation}
Let us also introduce homotopies of connections $\connsymp_t$, $\connbund_t$ and  corresponding homotopies of connections in Weyl bundles, ie. $\partial^\wbdot_t$ and $\gconn^\wbdot_t=g^{-1}_t \partial^\wbdot_t g_t$.  
Connections (\ref{connlocalformwbc}) and (\ref{connlocalformwbe}) can be rewritten in the uniform manner as $\partial_t^\wbdot=\rmd+\ibh[\Gamma_t^\wbdot \ccomm \cdot\,]$ with $\Gamma_t^\wbc=1/2 \Gamma_{ijk}(t)y^i y^j \rmd x^k$ and $\Gamma_t^\wbe=1/2 \Gamma_{ijk}(t)y^i y^j \rmd x^k-\rmi h \gambund(t)$. Denoting $\gend{\Gamma}_t^\wbdot=g_t^{-1} \Gamma_t^\wbdot$ one may check that
\begin{equation}
\label{gdergconn}
\gtder \gconn^\wbdot_t a = \gconn^\wbdot_t \gtder a +\ibh [\gtder \gend{\Gamma}^\wbdot_t \gcommt{t} a]
\end{equation}
Notice that unlike connection coefficients,  
$\frac{d}{dt} \Gamma_t^\wbdot$
define global, coordinate and frame independent sections of Weyl bundles. In turn the same stays true for $\gtder \gend{\Gamma}^\wbdot_t$.
The following theorem (5.4.3 of \cite{fedosovbook}) builds isomorphism theory for Fedosov algebras.
\begin{thrm}
\label{triv_trhm_liouv}
Let 
\begin{equation*}
D^\wbdot_t = \gconn^\wbdot_t +\ibh [\gamdot (t) \gcommt{t} \cdot\,]
\end{equation*}
be a homotopy of Abelian connections parameterized by $t \in [0,1]$, and let $\hdot(t)$ be a $t$-dependent $C$-section of $\wbdot$ (called Hamiltonian) satisfying $\mathrm{deg}(\hdot(t)) \geq 3$ and such that 
\begin{equation}
\label{hamiltoniancond}
D^\wbdot_t \hdot(t) - \gtder(\gend{\Gamma}_t^\wbdot+\gamdot(t))=\lambda(t)
\end{equation}
for some scalar $1$-form $\lambda(t)$.
Then, the equation
\begin{equation}
\label{triv_heis}
\gtder a+\frac{\rmi}{h}[\hdot(t)\gcommt{t} a]=0
\end{equation}
has the unique solution $a(t)$ for any given $a(0) \in \wbdot \otimes \Lambda$ and the mapping $a(0) \mapsto a(t)$ is isomorphism for any $t \in [0,1]$. Moreover, $a(0) \in \wbund^{D_0}_\dotmark$ if and only if $a(t) \in \wbund^{D_t}_\dotmark$.
\end{thrm}
\begin{proof}
Using definition of $\gtder$ and integrating (\ref{triv_heis}) one obtains
\begin{equation}
a(t)=g^{-1}_t g_0 a(0)-\ibh g^{-1}_t \int_0^t g_\tau [\hdot(\tau)\gcommt{\tau} a(\tau)] d\tau
\end{equation}
The operator $\ibh g^{-1}_t \int_0^t g_\tau [\hdot(\tau)\gcommt{\tau} \cdot\,]$ raises degree if $\mathrm{deg}(\hdot(t)) \geq 3$, and  defines unique iterative solution $a(t)$. Thus, the mapping $a(0) \mapsto a(t)$ is indeed bijective. From Leibniz rule we infer that $a(0) \gcirc_0 b(0) \mapsto a(t) \gcirc_t b(t)$, hence it is isomorphism. To verify its behavior on flat sections let us suppose that $a(t)$ is a solution of (\ref{triv_heis}) and calculate
\begin{multline*}
\gtder D^\wbdot_t a \stackrel{(\ref{gdergconn})}{=}\gconn^\wbdot_t \gtder a +\ibh [\gtder \gend{\Gamma}^\wbdot_t \gcommt{t} a]+
\ibh [\gtder  \gamdot \gcommt{t} a] + \ibh [  \gamdot \gcommt{t} \gtder a]\\
= D^\wbdot_t \gtder a +\ibh [\gtder (\gend{\Gamma}^\wbdot_t + \gamdot)\gcommt{t} a] \\
\stackrel{(\ref{triv_heis})}{=}
-\ibh [\hdot \gcommt{t} D^\wbdot_t a]+\ibh [\gtder (\gend{\Gamma}^\wbdot_t + \gamdot) -D^\wbdot_t \hdot \gcommt{t} a]
\end{multline*}
Thus, provided that $D^\wbdot_t \hdot -\gtder (\gend{\Gamma}^\wbdot_t + \gamdot)$ is a scalar $1$-form, $D^\wbdot_t a$ must be a solution of (\ref{triv_heis}). From uniqueness of solutions of (\ref{triv_heis}) one infers that $D^\wbdot_0 a(0)=0$ if and only if $D^\wbdot_t a(t)=0$.
\end{proof}
The very first observation concerning theorem \ref{triv_trhm_liouv} is that there is some consistency condition for $\lambda(t)$. Indeed, calculating $D^\wbdot_t D^\wbdot_t \hdot(t)$ one obtains from Abelian property of $D^\wbdot_t$ and (\ref{hamiltoniancond})
\begin{equation}
0=D^\wbdot_t D^\wbdot_t \hdot(t)=D^\wbdot_t \gtder(\gend{\Gamma}_t^\wbdot+\gamdot(t))+ \rmd\lambda(t)
\end{equation}
Straightforward but a bit longish calculation shows that 
\begin{equation}
\label{dconnon-coef}
D^\wbdot_t \gtder(\gend{\Gamma}_t^\wbdot+\gamdot(t))=\gtder \Omega^\wbdot_t=\frac{d}{dt} \Omega^\wbdot_t
\end{equation}
Hence
\begin{equation}
\label{dlambdacond}
\rmd\lambda(t) = - \frac{d}{dt} \Omega^\wbdot_t
\end{equation}
This observation results in the following theorem.
\begin{thrm}
\label{isocurvthrm}
Let $D^\wbdot_t$ be homotopy of Abelian connections
\begin{equation*}
D^\wbdot_t = \gconn^\wbdot_t +\ibh [\gamdot (t) \gcommt{t} \cdot\,]
\end{equation*}
A Hamiltonian generating isomorphism between corresponding algebras of flat sections exists iff there is some scalar $1$-form $\lambda(t)$ such that $\rmd\lambda(t) = - \frac{d}{dt} \Omega^\wbdot_t$.
%
\end{thrm}
\begin{proof} 
If a Hamiltonian exists, then consistency condition (\ref{dlambdacond}) holds as calculated before. Conversely if there is $\lambda(t)$ satisfying (\ref{dlambdacond}) then the Hamiltonian can be constructed as a solution of equation 
\begin{equation*}
D^\wbdot_t \hdot(t) = \gtder(\gend{\Gamma}_t^\wbdot+\gamdot(t))+\lambda(t)
\end{equation*}
Due to theorem \ref{triv_deqlemma} such solution exists because $D^\wbdot_t \Big(\gtder(\gend{\Gamma}_t^\wbdot+\gamdot(t))+\lambda(t)\Big)=0$ by relations (\ref{dconnon-coef}) and (\ref{dlambdacond}).
Thus, the Hamiltonian can be taken as 
\begin{equation*}
\hdot(t)= -{Q_{\dotmark}}_t \delta^{-1} \Big(\gtder(\gend{\Gamma}_t^\wbdot+\gamdot(t))+\lambda(t) \Big)
\end{equation*}
where ${Q_{\dotmark}}_t$ denotes quantization map corresponding to $D^\wbdot_t$.
\end{proof}
\begin{rmrk}
\label{alliso}
The immediate consequence of this theorem is that algebras generated by Abelian connections described in theorem \ref{fedo_abel} are isomorphic in the sense of theorem \ref{triv_trhm_liouv}, if and only if their curvatures are in the same cohomology class. Indeed, if they are isomorphic then condition (\ref{dlambdacond}) holds, which, after integration, yields that the curvatures must be cohomolgical. On the other hand, given two sets of input data for theorem \ref{fedo_abel} (ie{.} fiberwise product, symplectic connection, connection in the bundle, normalizing section $\mudot$ and curvature), we can always homotopically transform one to another and, in turn, obtain homotopy of Abelian connections. Thus, the only obstruction is the requirement (\ref{dlambdacond}), which can be satisfied for curvatures in the same cohomology class by homotopy $\Omega^\wbdot_t=\Omega^\wbdot_0  + t(\Omega^\wbdot_1 - \Omega^\wbdot_0)$.
\end{rmrk}

Since locally one can always find $\lambda(t)$ such that $\frac{d}{d t}\Omega^\wbdot_t=-d \lambda(t)$, then locally all Abelian connections of theorem \ref{fedo_abel} give rise to isomorphic algebras. In particular, all such algebras are locally isomorphic to the \emph{trivial algebra}, ie{.}  Moyal algebra, which is generated in theorem~\ref{fedo_abel}, by fiberwise Moyal product, $\kappa=0$, $\mudot=0$ and flat connections $\connsymp$, $\connbund$. The useful consequence of this fact is the following observation (compare \cite{fedosovbook}, corollary 5{.}5{.}2).
\begin{lmm}
\label{scalarcommlmm}
If $[s\gcomm a]=0$ for all $a\in \wbund^{D}_\dotmark$, then $s$ is some scalar form.
\end{lmm}
This fact can be easily verified for the trivial algebra. As a local statement it can be transported to arbitrary $\wbund^{D}_\dotmark$, because scalar forms remain unmodified by considered isomorphisms.

The next step is to observe that the mapping introduced in theorem \ref{triv_trhm_liouv} can be written in a bit more explicit form as
\begin{equation}
\label{uisodef}
T_t(a(0)):=a(t)=U_\dotmark^{-1}(t) \gcirc_t (g^{-1}_t g_0 a(0)) \gcirc_t U_\dotmark(t)
\end{equation}
where $U_\dotmark (t)$ is the solution of the equation
\begin{equation}
\label{udiffeq}
\gtder U_\dotmark (t) = \ibh U_\dotmark (t) \gcirc_t \hdot(t) 
\end{equation}
with $U_\dotmark (0)=1$. Notice that (\ref{udiffeq}) can be rewritten as 
\begin{equation}
U_\dotmark (t)=1+\ibh g^{-1}_t \int_0^t g_\tau (U_\dotmark (\tau) \gcirc_\tau \hdot(\tau)) d\tau
\end{equation}
and by iteration method it follows that $U_\dotmark (t)$ is defined uniquely. However, the operator $\ibh g^{-1}_t \int_0^t g_\tau ((\cdot) \gcirc_\tau \hdot(\tau)) d\tau$ may produce negative powers\footnote{Strictly, this means that $U_\dotmark (t)$ belongs to what is called \emph{extended Weyl bundle} in \cite{fedosovbook}. As the notation of the present paper seems to be cumbersome enough without precise dealing with this issue, we are not going to introduce any separate symbol for those extended bundles.} of $h$, but these terms do not introduce negative total degree (\ref{degdef}), and appear in finite number for each total degree, thus we avoid problem of infinite series of monomials at fixed total degree. The inverse $U_\dotmark^{-1}(t)$ is taken with respect to $\gcirc_t$ and it is uniquely determined by the equation
\begin{equation}
\label{uinvdiffeq}
\gtder U^{-1}_\dotmark (t) = - \ibh  \hdot(t) \gcirc_t U^{-1}_\dotmark (t)  
\end{equation}
Notice that composing two mappings of type (\ref{uisodef}) one obtains a mapping of the same type, ie{.} for $U_\dotmark$ and $U'_\dotmark$ defining respectively mappings from $\gcirc_0$ to $\gcirc_1$ and from $\gcirc_1$ to $\gcirc_2$ their composition yields
\begin{equation}
\label{ucomp}
U_\dotmark^{\prime -1}\ \gcirc_2 \bigg(g^{-1}_2 g_1 
\Big(U_\dotmark^{ -1} \gcirc_1 (g^{-1}_1 g_0 a) \gcirc_1 U_\dotmark\Big)\bigg) 
\gcirc_2 U'_\dotmark
=\big(U''_\dotmark\big)^{-1} \gcirc_2 (g^{-1}_2 g_0 a) \gcirc_2 U''_\dotmark
\end{equation}
with $U''_\dotmark=\Big((g^{-1}_2 g_1 U_\dotmark) \gcirc_2 U'_\dotmark \Big)$ describing mapping from $\gcirc_0$ to $\gcirc_2$ and $\big(U''_\dotmark\big)^{-1}$ being inverse with respect to $\gcirc_2$. It follows that the inverse mapping $T_t^{-1}$ can be written  as
\begin{equation}
\label{uinvisodef}
T^{-1}_t(a)=\big(g_0^{-1} g_t U_\dotmark (t)\big) \gcirc_0 (g^{-1}_0 g_t a) \gcirc_0 \big(g_0^{-1} g_t U^{-1}_\dotmark(t)\big)
\end{equation}
where $U_\dotmark (t)$ and $U_\dotmark^{-1}(t)$ are the same quantities that in (\ref{uisodef}).
\begin{rmrk}
\label{invhmtprmrk}
For technical reasons which become clear later, we would like to emphasize the following fact. Let $T_t$ be homotopy of isomorphisms generated by homotopy $D^\wbdot_t$, Hamiltonian $\hdot(t)$ and scalar form $\lambda(t)$. The inverse $T_1^{-1}$ of an ``endpoint'' $T_1$ can be always represented as an ``endpoint'' $T'_1=T_1^{-1}$ of homotopy $T'_t$ generated by $D^{\prime\wbdot}_t=D^\wbdot_{1-t}$, Hamiltonian $\hdot'(t)=-\hdot(1-t)$ and  $\lambda'(t)=-\lambda(1-t)$. To verify that $T'_t$ is well defined, we consequently put $\gend{\Gamma}_t^{\prime\wbdot}=\gend{\Gamma}_{1-t}^{ \wbdot}$, $\gamdot'(t)=\gamdot(1-t)$, $g'_t=g_{1-t}$, $\gtder'=g^{\prime\, -1}_t \frac{d}{dt} g'_t$ and let $\gcirc'_t$ stand for $\gcirc_{1-t}$. Then $D^{\prime\wbdot}_t \hdot'(t) = \gtder'(\gend{\Gamma}_t^{\prime\wbdot}+\gamdot'(t))+\lambda'(t)$ holds as a result of analogous relation for $T_t$. Moreover, one can observe that $U'_\dotmark (t)=(g_t^{\prime -1}g_1 U^{-1}_\dotmark(1)) \gcirc'_t U_\dotmark (1-t)$ is the solution of $\gtder' U'_\dotmark (t) = \ibh U'_\dotmark (t) \gcirc'_t \hdot'(t)$, satisfying $U'_\dotmark(0)=1$, provided that $U_\dotmark(t)$ is the solution of corresponding equation for $T_t$. Hence, by (\ref{uinvisodef}), the relation $T'_1(a)=U^{\prime -1}_\dotmark (1) \gcirc'_1 (g^{\prime -1}_1 g'_0 a) \gcirc'_1 U'_\dotmark (1)=(g_0^{-1}g_1 U_\dotmark(1)) \gcirc_0 (g^{-1}_0 g_1 a) \gcirc_0 (g_0^{-1}g_1 U^{-1}_\dotmark(1))=T^{-1}_1 (a)$ indeed holds\footnote{Here $U^{\prime -1}_\dotmark (1)$ is inverese with respect to $\gcirc'_1$, ie. $\gcirc_0$, while $U^{-1}_\dotmark (1)$ is inverese for $\gcirc_1$.}.
\end{rmrk}

Using isomorphism $T_t$ one can push-forward arbitrary connection $\nabla^\wbdot$ (defined for fiberwise product $\gcirc_0$) \begin{equation}
\label{pushDdef}
T_{t*}\nabla^\wbdot = T_t \nabla^\wbdot T_t^{-1}
\end{equation}
The immediate observation coming directly from the definition (\ref{pushDdef}) and properties of $T_t$ is that $T_{t*}\nabla^\wbdot$ is $+1$-derivation with respect to $\gcirc_t$ and also an Abelian connection if $\nabla^\wbdot$ is such. Inserting (\ref{uisodef}) and (\ref{uinvisodef}) into (\ref{pushDdef}) one can compute that
\begin{equation}
\label{pushed_conn_form1}
T_{t*}\nabla^\wbdot = g^{-1}_t g_0 \nabla^\wbdot g^{-1}_0 g_t + [U^{-1}_\dotmark(t) \gcirc_t \big(g^{-1}_t g_0 \nabla^\wbdot g^{-1}_0 g_t U_\dotmark(t)\big) \gcommt{t} \cdot\, ]
\end{equation}
Let $\gend{\rmd}_t= g^{-1}_t \rmd g_t$ and let
\begin{equation}
\nabla^\wbdot= \gend{\rmd}_0 + \ibh [\Upsilon^\wbdot \gcommt{0}\cdot\,]
\end{equation}
be local representation of $\nabla^\wbdot$. Since
\begin{equation}
g^{-1}_t g_0 \nabla^\wbdot g^{-1}_0 g_t= \gend{\rmd}_t + \ibh[g_t^{-1} g_0\Upsilon^\wbdot \gcommt{t} \cdot\,]
\end{equation}
it is possible to write down the following local formula
\begin{equation}
T_{t*}\nabla^\wbdot = \gend{\rmd}_t + \ibh [g_t^{-1} g_0 \Upsilon^\wbdot -\rmi h U^{-1}_\dotmark(t) \gcirc_t \big(g^{-1}_t g_0 \nabla^\wbdot g^{-1}_0 g_t U_\dotmark(t)\big) \gcommt{t} \cdot\, ]
\end{equation}
and we are justified to define local push-forward of $\Upsilon^\wbdot$
\begin{equation}
\label{pushedcoefdef}
T_{t*}\Upsilon^\wbdot= g_t^{-1} g_0 \Upsilon^\wbdot - \rmi h U^{-1}_\dotmark(t) \gcirc_t \big(g^{-1}_t g_0 \nabla^\wbdot g^{-1}_0 g_t U_\dotmark(t)\big)
\end{equation}
Suppose that $T$ maps form $\gcirc_0$ to $\gcirc_1$ and $T'$ maps form $\gcirc_1$ to $\gcirc_2$. Then
\begin{equation}
T'_*T_*\nabla^\wbdot =(T'T)_* \nabla^\wbdot
\end{equation}
and also, after some calculations, one may check that
\begin{equation}
\label{mltplpshes}
T'_*T_*\Upsilon^\wbdot =(T'T)_* \Upsilon^\wbdot
\end{equation}
Clearly, we are able to rewrite Abelian connections $D^\wbdot_t$ in the local form $D^\wbdot_t= \gend{\rmd}_t + \ibh [ \gend{\Gamma}_t^\wbdot+ \gamdot (t)\gcommt{t}\cdot\,]$.
With above notations an useful lemma on push-forward of $D^\wbdot_0$  can be formulated (previously considered in~\cite{dobrski-sw}).
\begin{lmm}
\label{pushedconnslmm}
Connections $T_{t*}D^\wbdot_0$ and $D^\wbdot_t$ coincide. Moreover if $T_t$ is generated by Hamiltonian satisfying $D^\wbdot_t \hdot(t) - \gtder(\gend{\Gamma}_t^\wbdot+\gamdot(t))=\lambda (t)$ with scalar $1$-form $\lambda(t)$ then 
\begin{equation}
\label{pushedconnsectrel}
T_{t*}\Big(\gend{\Gamma}_0^\wbdot + \gamdot (0) \Big) =  \gend{\Gamma}_t^\wbdot + \gamdot (t) + \int_0^t \lambda (\tau) d\tau
\end{equation}
\end{lmm}
\begin{proof}
From theorem \ref{triv_trhm_liouv} it follows that subalgebras of flat sections are the same for both considered connections, ie{.} $\wbund^{D_t}_\dotmark=\wbund^{T_{t*}D_0}_\dotmark$. Hence, for arbitrary $a \in \wbund^{D_t}_\dotmark$ it holds that $T_{t*}D^\wbdot_0 a=D^\wbdot_t a=0$. Locally this relation yields
\begin{equation}
[T_{t*}\Big(\gend{\Gamma}_0^\wbdot + \gamdot (0) \Big) \gcommt{t} a]=[\gend{\Gamma}_t^\wbdot + \gamdot (t)\gcommt{t} a]
\end{equation}
for all $a \in \wbund^{D_t}_\dotmark$. Thus, by lemma \ref{scalarcommlmm}, $T_{t*}\Big(\gend{\Gamma}_0^\wbdot + \gamdot (0) \Big)$ and $\gend{\Gamma}_t^\wbdot + \gamdot (t)$ differ by some scalar form, and consequently $T_{t*}D^\wbdot_0 a=D^\wbdot_t a$ for arbitrary, not necessarily flat, section~$a$. Using (\ref{udiffeq}), (\ref{uinvdiffeq}) and (\ref{pushed_conn_form1}) we calculate
\begin{multline*}
\gtder T_{t*}\Big(\gend{\Gamma}_0^\wbdot + \gamdot (0) \Big) \\
= -\rmi h \gtder U^{-1}_\dotmark(t) \gcirc_t \big(g^{-1}_t g_0 D^\wbdot_0 g^{-1}_0 g_t U_\dotmark(t)\big) -\rmi h  U^{-1}_\dotmark(t) \gcirc_t \big(g^{-1}_t g_0 D^\wbdot_0 g^{-1}_0 g_t \gtder U_\dotmark(t)\big)\\
=[U^{-1}_\dotmark(t) \gcirc_t \big(g^{-1}_t g_0 D^\wbdot_0 g^{-1}_0 g_t U_\dotmark(t)\big) \gcommt{t} \hdot(t)] +
g^{-1}_t g_0 D^\wbdot_0 g^{-1}_0 g_t \hdot(t) \\
= T_{t*}D^\wbdot_0 \hdot(t)=D^\wbdot_t \hdot(t)=\gtder(\gend{\Gamma}_t^\wbdot+\gamdot(t)) + \lambda (t)
\end{multline*}
and the formula (\ref{pushedconnsectrel}) can be obtained by application of $g^{-1}_t \int_0^t g_\tau \, \cdot \, d \tau$ to above relation.
\end{proof}
Now we are about formulating quite important fact concerning automorphisms of Fedosov algebras. (The lemma given below is essentially generalized variant of proposition 5{.}5{.}5 of \cite{fedosovbook}. Here, we reproduce it with different proof, which employs lemma \ref{pushedconnslmm}, and explicitly relates function $f$, used for rescaling, to $1$-forms $\lambda^{(i)}$ corresponding to involved isomorphisms).
\begin{lmm}
\label{locadjlmm}
Let $D^\wbdot_0,\dots D^\wbdot_{n-1}$ be Abelian connections defined with respect to fiberwise products $\gcirc_0,\dots,\gcirc_{n-1}$ for $n \geq 1$. Consider homotopies of isomorphisms $T_t^{(i)}$, such that \linebreak $T_1^{(1)}: \wbund^{D_0}_\dotmark \to \wbund^{D_1}_\dotmark, T_1^{(2)}: \wbund^{D_1}_\dotmark \to \wbund^{D_2}_\dotmark, \dots, T_1^{(n)}: \wbund^{D_{n-1}}_\dotmark \to \wbund^{D_0}_\dotmark$. Their composition  $T_1^{(n)} \dots T_1^{(1)}$ can be always locally represented as
\begin{equation}
T_1^{(n)} \dots T_1^{(1)} a = V^{-1}_\dotmark \gcirc_0 a \gcirc_0 V_\dotmark
\end{equation}
for arbitrary $a \in \wbdot$, and with $V_\dotmark$ being flat section belonging to $\wbund^{D_0}_\dotmark$.
\end{lmm}
\begin{proof}
For sake of more compact notation let us define $D^\wbdot_n=D^\wbdot_0$ and let $\gcirc_n$ stand for $\gcirc_0$. To construct $V_\dotmark$ we proceed as follows. Isomorphisms under consideration can be written as $T_1^{(i)}a = 
\big(U^{(i)}_\dotmark\big)^{-1} \gcirc_i a \gcirc_i U^{(i)}_\dotmark$. For their composition formula (\ref{ucomp}) yields $T_1^{(n)} \dots T_1^{(1)} a = U^{-1}_\dotmark \gcirc_0 a \gcirc_0 U_\dotmark$ with
\begin{equation*}
U_\dotmark=g_0^{-1}g_1 U^{(1)}_\dotmark \gcirc_0 \dots \gcirc_0 g_0^{-1}g_{n-1} U^{(n-1)}_\dotmark \gcirc_0 U^{(n)}_\dotmark
\end{equation*}
We are going to show that $U_\dotmark$ can be locally rescaled by some scalar function and become a flat section with respect to $D^\wbdot_0$. For this purpose lemma \ref{pushedconnslmm} can be used. First, one immediately gets that $\big(T_1^{(i)} \dots T_1^{(1)}\big)_* D^\wbdot_0=D^{\wbdot}_i$. A bit less straightforward relation holds for local coefficients of these connections. Let $D^\wbdot_i= \gend{\rmd}_i + \ibh [ \gend{\Gamma}_i^\wbdot+ \gamma_{\dotmark \, i} \gcommt{i}\cdot\,]$. Using definition (\ref{pushedcoefdef}) we obtain
\begin{equation}
\label{autolmmcalc1}
\Big(T_1^{(n)} \dots T_1^{(1)}\Big)_* \Big(\gend{\Gamma}_0^\wbdot + \gamma_{\dotmark \, 0} \Big)=\gend{\Gamma}_0^\wbdot + \gamma_{\dotmark \, 0}  - \rmi h U^{-1}_\dotmark \gcirc_0 \big( D^\wbdot_0  U_\dotmark \big)
\end{equation}
On the other hand, since each isomorphism $T_t^{(i)}$ is generated by some Hamiltonian $H_{\dotmark\,i}(t)$ satisfying condition (\ref{hamiltoniancond}) with some scalar $1$-form $\lambda^{(i)} (t)$, lemma \ref{pushedconnslmm} can be consecutively applied
\begin{multline}
\label{autolmmcalc2}
T_{1*}^{(n)} \dots T_{1*}^{(2)} T_{1*}^{(1)} \Big(\gend{\Gamma}_0^\wbdot + \gamma_{\dotmark \, 0} \Big)= T_{1*}^{(n)} \dots T_{1*}^{(2)}\Big(\gend{\Gamma}_1^\wbdot + \gamma_{\dotmark \, 1} + \int_0^1 \lambda^{(1)} (\tau) d\tau\Big)\\
=\dots = T_{1*}^{(n)}\Big(\gend{\Gamma}_{n-1}^\wbdot + \gamma_{\dotmark \, n-1} + \int_0^1 \lambda^{(1)} (\tau) d\tau + \dots + \int_0^1 \lambda^{(n-1)} (\tau) d\tau\Big)\\
=\gend{\Gamma}_{0}^\wbdot + \gamma_{\dotmark \, 0} + \int_0^1 \lambda^{(1)} (\tau) d\tau + \dots + \int_0^1 \lambda^{(n)} (\tau) d\tau
\end{multline}
In virtue of (\ref{mltplpshes}) we can combine (\ref{autolmmcalc1}) with (\ref{autolmmcalc2}) and obtain
\begin{equation}
\label{autolmmcalc3}
- \rmi h U^{-1}_\dotmark \gcirc_0 \big( D^\wbdot_0  U_\dotmark \big)= \int_0^1 \lambda^{(1)} (\tau) d\tau + \dots + \int_0^1 \lambda^{(n)} (\tau) d\tau
\end{equation}
The right hand side of above formula is a closed $1$-form. Indeed, integration of relation (\ref{dlambdacond}) for each $\lambda^{(i)}$ gives
$\rmd \int_0^1 \lambda^{(i)} (\tau) d\tau = -\Omega^\wbdot_i + \Omega^\wbdot_{i-1}$, where $\Omega^\wbdot_i$ is the curvature of $D^\wbdot_i$. Summing all these terms produces
\begin{equation}
\rmd \Bigg( \int_0^1 \lambda^{(1)} (\tau) d\tau + \dots + \int_0^1 \lambda^{(n)} (\tau) d\tau \Bigg) = 0
\end{equation}
because $\Omega^\wbdot_n=\Omega^\wbdot_0$. Thus (\ref{autolmmcalc3}) can be always locally rewritten as 
\begin{equation}
\label{dconnucond}
 D^\wbdot_0  U_\dotmark =\ibh U_\dotmark  \rmd f
\end{equation}
for some scalar function $f$, such that $\rmd f=\int_0^1 \lambda^{(1)} (\tau) d\tau + \dots + \int_0^1 \lambda^{(n)} (\tau) d\tau$. Finally, we define\footnote{We should comment on the lack of negative powers of $h$ in the term $e^{-\ibh f}$. This fact comes quite easily from conditions $\deg(\hdot(t)) \geq 3$ and (\ref{hamiltoniancond}). Indeed, analyzing right hand side of relation  $\lambda(t)=D^\wbdot_t \hdot(t) - \gtder(\gend{\Gamma}_t^\wbdot+\gamdot(t))$ we observe that the term $\gtder(\gend{\Gamma}_t^\wbdot+\gamdot(t))$ does not introduce monomials of degree 0 if $\deg(\gamdot(t)) \geq 1$. (We want $\gamdot(t)$ to be a $C$-section. Hence, there could be only scalar terms of degree $0$, but such terms does not affect $D^\wbdot_t$ and can be safely omitted). The same stays true for the term $D^\wbdot_t \hdot(t)$ provided that $\deg(\hdot(t)) \geq 3$. Thus $\lambda(t)$ must be a scalar form of degree 2 or greater and consequently $f$ can be taken without zeroth power of $h$.
}
(still locally) $V_\dotmark=e^{-\ibh f} U_\dotmark$.
From (\ref{dconnucond}) it follows that $D^\wbdot_0 V_\dotmark=0$, while 
\begin{equation}
V_\dotmark^{-1} \gcirc_0 a \gcirc_0 V_\dotmark = U_\dotmark^{-1} \gcirc_0 a \gcirc_0 U_\dotmark = T_1^{(n)} \dots T_1^{(1)} a
\end{equation}
\end{proof}
Let us additionally observe, that -- by remark \ref{invhmtprmrk} -- in above lemma one (or more) isomorphism $T_1^{(i)}: \wbund^{D_{i-1}}_\dotmark \to \wbund^{D_i}_\dotmark$ can be safely replaced by some $T^{\prime -1}_1$, provided that $T'_1: \wbund^{D_{i}}_\dotmark \to \wbund^{D_{i-1}}_\dotmark$.
\subsection{Seiberg-Witten map}
As a quite direct application of above considerations on isomorphisms of Fedosov algebras one can make the following observation concerning Seiberg-Witten map \cite{seibwitt}, \cite{jurco1}, \cite{dobrski-sw}. Suppose that there is some fiberwise adjoint Lie group action in the endomorphism bundle, realized by invertible elements of the fiber, ie. one is dealing with $\adjact_x : \ebund_x \to \ebund_x$ given smoothly for all $x \in \sympman$ and acting by
\begin{equation*}
\adjact_x (A)=\mathfrak{g} A \mathfrak{g}^{-1}
\end{equation*}
for both $A$ and $\mathfrak{g}$ belonging to the fiber $\ebund_x$. Action of $\adjact$ can be extended naturally to the Weyl bundle $\wbe$. Let $D^\wbe$ be arbitrary Abelian connection with curvature $\Omega^\wbe$, and let $D^\wbe_T$ be some local Abelian connection with the same curvature, but originating in \emph{flat} connection $\connbund_T$ in the bundle $\bund$. Due to Fedosov isomorphism theory, one is able to establish local isomorphism $T$ mapping $\wbund^{D}_\ebund$ to $\wbund^{D_T}_\ebund$, with $\lambda(t)=0$ in theorem \ref{triv_trhm_liouv}. On the other hand, using gauge transformation $\adjact$, the Abelian connection 
\begin{equation}
{D^\wbe}'=\adjact_* D^\wbe = \adjact D^\wbe \adjact^{-1} 
\end{equation}
is obtained. It can be easily observed, that if $D^\wbe$ comes from theorem \ref{fedo_abel} for connection $\connend$ and section $\mue$, then ${D^\wbe}'$ also originates in this theorem with ${\Omega^\wbe}'=\Omega^\wbe$, but for the gauge transformed objects ${\connend}'= \adjact \connend \adjact^{-1}$ and ${\mue}'=\adjact \mue$. For ${D^\wbe}'$ one can also set up local isomorphism $T'$ mapping $\wbund^{D'}_\ebund$ to $\wbund^{D_T}_\ebund$, again with $\lambda'(t)=0$. Then the following question arises: how $T$ and $T'$ are related to each other?

To answer it one can use, \emph{mutatis mutandis}, lemma \ref{locadjlmm}. Consider composition $T' \adjact T^{-1}$ which maps $\wbund^{D_T}_\ebund$ to itself. We need to represent $\adjact$ as an endpoint of some homotopy of isomorphisms. Clearly, this cannot be homotopy generated strictly by theorem \ref{triv_trhm_liouv}. Notice however that if $\mathfrak{g}$ belongs to connected component of identity, ie{.} it can be written as $\mathfrak{g}=e^{\alpha}$ for some $\alpha \in C^{\infty}(\sympman,\ebund)$, then one can introduce homotopy $\adjacthmtp$ which acts on sections of $\wbe$ by $\adjacthmtp (a) = e^{\alpha t} a e^{-\alpha t}$. (Of course, the fiberwise product $\gcirc$ is constant with respect to $t$ here). This homotopy is generated by Hamiltonian $H=\rmi h \alpha$ of degree 2 and equation $\gtder a+\ibh [H \gcomm a]=0$ in analytic (not formal) sense. The relation (\ref{hamiltoniancond}) holds true with $D^\wbe_t=D^\wbe+[e^{\alpha t}D^\wbe e^{-\alpha t}\gcomm \cdot\,]$ and $\lambda(t)=0$. Also, the formula (\ref{pushedconnsectrel}) remains valid, thus we are able to repeat reasoning of lemma \ref{locadjlmm} with $\adjacthmtp$ put in place of some $T^{(i)}_t$. In particular, the local relation 
\begin{equation*}  
T' \adjact T^{-1} a =
V_\ebund^{-1} \gcirc a \gcirc V_\ebund 
\end{equation*}
holds true for $V_\ebund \in \wbund^{D_T}_\ebund$ and arbitrary $a \in \wbdot$. This however implies that
\begin{equation}
\label{sw_in_weylbund}  
T' \adjact a =
V_\ebund^{-1} \gcirc T a \gcirc V_\ebund 
\end{equation}
Projecting (\ref{sw_in_weylbund}) back to $\ebund$ and using relation $\adjact Q_\ebund = Q'_\ebund \adjact$ one obtains for $A \in C^\infty(\sympman,\ebund)$ and $A'=\adjact A= \mathfrak{g}A\mathfrak{g}^{-1}$
\begin{equation}
\label{sw_in_ebund}  
M' (A') =
G *_T M(A) *_T G^{-1} 
\end{equation}
where $M=Q_{T\,\ebund}^{-1} T Q_\ebund$, $M'=Q_{T\,\ebund}^{-1} T' Q'_\ebund$ and $G=Q_{T\,\ebund}^{-1} V_\ebund^{-1}$ with $Q_\ebund$, $Q'_\ebund$, $Q_{T\,\ebund}$ corresponding to $D^\wbe$, ${D^\wbe}'$ and $D^\wbe_T$ respectively. But (\ref{sw_in_ebund}) says that the covariance relation for map $M$ is exactly the one which should occur for Seiberg-Wiiten map applied on endomorphisms. Indeed, if one chooses frame in $\bund$ for which connection coefficients of $\connbund_T$ vanish, then $*_T$ acts on matrices representing endomorphisms in this frame according to the usual ``row-column'' matrix multiplication rule, but with commutative product of entries replaced by noncommutative Fedosov star product of functions (which can be chosen to be Moyal product for flat $\connsymp$). Notice also, that the leading term of $G$ is given by $\mathfrak{g}$ as it comes from our application of modified lemma \ref{locadjlmm}. Thus, one is able to rewrite (\ref{sw_in_ebund}) in more conventional notation as
\begin{equation}
\label{sw_in_ebund_conv}  
\widehat{A}' =
\widehat{\mathfrak{g}} *_T \widehat{A} *_T \widehat{\mathfrak{g}}^{-1} 
\end{equation}
with hat denoting appropriate Seiberg-Witten map. One can reproduce usual Seiberg-Witten equations for gauge objects (gauge potential and field strength) from formula (\ref{sw_in_ebund}). For this kind of considerations, as well as examples of explicit calculations of Seiberg-Witten map, the reader is referred to $\cite{dobrski-sw}$. Here, let us close this subsection with the statement that in our generalized context, Seiberg-Witten map again appeared as a local isomorphism, which is just some property of global Fedosov quantization of endomorphism bundle.  


\section{Trace functional}
We are ready to define trace functional for generalized Fedosov algebras. Fortunately, we already know that the trace exists for the case of fiberwise Moyal product, ie. for the case of original formulation of \cite{fedosovbook}. This fact can be used to avoid difficulties related to gluing together local trivializations. However, the isomorphisms theory of previous section remains crucial. Thus, in what follows the term ``isomorphism'' will always refer to isomorphism described by theorem \ref{triv_trhm_liouv}. Again, in this section $\dotmark$ stands for $\mathbb{C}$ or $\ebund$.

For arbitrary Abelian connection $D^\wbdot$ generated by theorem \ref{fedo_abel} with fiberwise product $\gcirc$, symplectic connection $\connsymp$, connection in the bundle $\connbund$, normalizing section $\mudot$ and curvature $\Omega^\wbdot$, let $D^\wbdot_F$ denote Abelian connection obtained for fiberwise Moyal product $\circ$ and remaining data unchanged. Thus $D^\wbdot_F$ is Abelian connection in the sense of \cite{fedosovbook}, generated by the first Fedosov theorem, and $Q_{\dotmark F}$ is corresponding quantization map. For a compactly supported section $a \in \wbund^{D_F}_\dotmark$ its Fedosov trace will be denoted $\tr_F a$. Recall that for the trivial (ie{.} Moyal) algebras the trace is just integral\footnote{Here we omit normalizing constant $(2\pi h)^{-n}$ in front of integral, as compared to \cite{fedosovbook}.} $\tr_M a=\int_{\mathbb{R}^{2n}} Q_{\dotmark M}^{-1}(a) \frac{\omega^n}{n!}$. For general case one introduces local isomorphisms $T_i$ to trivial algebra and corresponding partition of unity $\rho_i$. The trace is then defined as $\tr_F a =\sum_i \tr_M T_i(Q_{\dotmark F}(\rho_i) \circ a)$. One can show that such definition depends neither on the choice of $T_i$ nor $\rho_i$. The important properties of the trace $\tr_F$ are that
\begin{equation}
\label{tracecomm}
\tr_F(a\circ b)=\tr_F(b\circ a)
\end{equation}
for all $a,b \in  \wbund^{D_F}_\dotmark$, and that
\begin{equation}
\label{traceiso}
\tr'_{F} T a = \tr_F a
\end{equation}
provided that $T$ maps one Fedosov algebra $\wbund^{D_F}_\dotmark$ with trace $\tr_F$, to another $\wbund^{D'_F}_\dotmark$ with trace $\tr'_{F}$. Notice however, that as to this point $T$ must denote isomorphism which results from theorem \ref{triv_trhm_liouv} with constant homotopy of fiberwise Moyal products. This is because (\ref{traceiso}) has been proven in \cite{fedosovbook} only in such context. We are interested in properties of the trace under present, broader class of isomorphisms. With following lemma all the required relations can be easily obtained.
\begin{lmm}
\label{autotracelmm}
Let $D^\wbdot_F,D^\wbdot_1,\dots D^\wbdot_{n-1}$ be Abelian connections defined with respect to fiberwise products $\circ,\gcirc_1,\dots,\gcirc_{n-1}$ for $n \geq 1$. Consider homotopies of isomorphisms $T_t^{(i)}$, such that \linebreak $T_1^{(1)}: \wbund^{D_F}_\dotmark \to \wbund^{D_1}_\dotmark, T_1^{(2)}: \wbund^{D_1}_\dotmark \to \wbund^{D_2}_\dotmark, \dots, T_1^{(n)}: \wbund^{D_{n-1}}_\dotmark \to \wbund^{D_F}_\dotmark$. For their composition  $T_1^{(n)} \dots T_1^{(1)}$ and compactly supported $a \in \wbund^{D_F}_\dotmark$ the global relation
\begin{equation}
\tr_F T_1^{(n)} \dots T_1^{(1)} a = \tr_F a
\end{equation}
holds.
\end{lmm}
\begin{proof}
Clearly, $ T_1^{(n)} \dots T_1^{(1)}$ maps $\wbund^{D_F}_\dotmark$ to $\wbund^{D_F}_\dotmark$, hence lemma \ref{locadjlmm} can be applied. Let $\{\mathcal{O}_i \}$ be covering of the support of section $a$ (by compactness one can make this covering finite), such that in each $\mathcal{O}_i$ lemma \ref{locadjlmm} holds with some $V_{\dotmark i} \in \wbund^{D_F}_\dotmark$, and choose some compatible partition of unity $\sum_i \rho_i=1$. Consequently $\sum_i Q_{\dotmark F}(\rho_i)=1$. Then 
\begin{multline}
\tr_F T_1^{(n)} \dots T_1^{(1)} a=\tr_F  T_1^{(n)} \dots T_1^{(1)}\Big(\sum_i Q_{\dotmark F}(\rho_i) \circ a\Big)\\
=\sum_i\tr_F \Big( V^{-1}_{\dotmark i} \circ Q_{\dotmark F}(\rho_i) \circ a \circ V_{\dotmark i}\Big)
=\sum_i\tr_F \Big( Q_{\dotmark F}(\rho_i) \circ a \Big) = \tr_F a
\end{multline}
where linearity of $\tr_F$ and property (\ref{tracecomm}) have been used.
\end{proof}
Again, by remark \ref{invhmtprmrk}, above lemma covers also the case of appropriate inverses put in place of some isomorphisms $T_1^{(i)}$s. Now, let $T$ be ``endpoint'' of some homotopy of our generalized isomorphisms (ie. $\gcirc_t$ does not have to be equal to $\circ$ for each $t$) mapping $\wbund^{D_F}_\dotmark$ with trace $\tr_F$, to $\wbund^{D'_F}_\dotmark$ with trace $\tr'_{F}$. Then, there must exist isomorphism $T' : \wbund^{D_F}_\dotmark \mapsto \wbund^{D'_F}_\dotmark$ being ``endpoint'' of homotopy for which $\gcirc_t$ is constantly equal to Moyal product\footnote{This is because $D^\wbdot_F$ and $D^{\prime\wbdot}_F$ are generated by some input data for theorem \ref{fedo_abel}, but both with fiberwise Moyal product. We can homotopically tranform one set of such input data into another without modyfing fiberwise product and generate homotopy of Abelian connections, each with respect to $\circ$. Then, by theorem \ref{isocurvthrm}, the compatible Hamiltonian must exists since we already know that condition (\ref{dlambdacond}) can be fulfilled due to existence of $T$.}. Applying (\ref{traceiso}) for $T'$ and lemma \ref{autotracelmm} for $T^{\prime -1} T$ one obtains
\begin{equation}
\tr'_F T a =\tr'_F T' T^{\prime -1} Ta=\tr_F T^{\prime -1} Ta= \tr_F a
\end{equation}
Thus, relation (\ref{traceiso}) holds true for generalized isomorphisms.
 
\begin{dfn}
For arbitrary $\wbund^{D}_\dotmark$, corresponding to Abelian connection $D^\wbdot$ of theorem \ref{fedo_abel}, the trace of compactly supported $a \in \wbund^{D}_\dotmark$ is given by
\begin{equation}
\label{trdef}
\tr a =\tr_F T(a)
\end{equation}
where $T$ is arbitrary isomorphism mapping $\wbund^{D}_\dotmark$ to $\wbund^{D_F}_\dotmark$.
\end{dfn}
First, algebras $\wbund^{D}_\dotmark$ and $\wbund^{D_F}_\dotmark$ are indeed isomorphic, as guaranteed by remark \ref{alliso}. We should also check that above definition does not depend on particular choice of isomorphism $T$. Let $T'$ be some other isomorphism mapping $\wbund^{D}_\dotmark$ to $\wbund^{D_F}_\dotmark$.  Clearly, $T' T^{-1}$ maps $\wbund^{D_F}_\dotmark$ to $\wbund^{D_F}_\dotmark$, hence lemma \ref{autotracelmm} can be applied (together with remark \ref{invhmtprmrk}) yielding
\begin{equation}
\label{trdefcalc}
\tr_F T'(a)=\tr_F  T' T^{-1} T(a) = \tr_F T(a)
\end{equation}
Similar argument can be employed to show that relation (\ref{traceiso}) can be carried over to generalized trace. 
Indeed, consider arbitrary isomorphism $T_X$ mapping $\wbund^{D}_\dotmark$ with trace defined by $\tr a = \tr_F T(a)$ to $\wbund^{D'}_\dotmark$ with trace $\tr' a = \tr'_F T'(a)$. Since $\wbund^{D}_\dotmark$ and $\wbund^{D'}_\dotmark$ are isomorphic, then also $\wbund^{D_F}_\dotmark$ and $\wbund^{D'_F}_\dotmark$ are isomorphic, because, by definition, connections $D^\wbdot_F$ and $D^{\prime\wbdot}_F$  inherit curvatures from their generalized counterparts. Thus remark \ref{alliso} can be applied, and let $T_F : \wbund^{D_F}_\dotmark \to \wbund^{D'_F}_\dotmark$ be corresponding isomorphism. Then
\begin{equation}
\tr' T_X a = \tr'_F T' T_X a =\tr'_F T_F T_F^{-1} T' T_X T^{-1} T a= \tr_F  T_F^{-1} T' T_X T^{-1} T a=\tr_F T a=\tr a 
\end{equation}
where we have first used (\ref{traceiso}) for $T_F$, and then lemma \ref{autotracelmm} for $T_F^{-1} T' T_X T^{-1}$.
Finally, the relation
\begin{equation}
\label{tracegencomm}
\tr(a\gcirc b)=\tr(b\gcirc a)
\end{equation}
for compactly supported sections $a,b \in  \wbund^{D}_\dotmark$ is the direct consequence of (\ref{tracecomm}).
\section{Some explicit expressions}
Our goal for this section is to give explicit expressions, up to second power of deformation parameter, for generalized star products and trace functional. They appear to be quite bulky, and calculations leading to them are cumbersome. In fact, it would be overwhelmingly tedious task to perform them manually. Instead, \emph{Mathematica} system has been used together with useful tensor manipulation package -- \emph{xAct} \cite{xact}. The interested reader can find the \emph{Mathematica} file relevant for these calculations on the author's website \cite{mathemfile}. This file also contains ``intermediate'' structures (Abelian connection, lifting to Weyl bundle, trivialization isomorphism), as well as some tests for validity of obtained results (including associativity of star product and isomorphicity of trivialization). The results are given for arbitrary fiberwise isomorphism (\ref{gdef}), but are restricted to the case of vanishing normalizing section $\mudot=0$, and curvature corrections $\kappa=0$. The main reason for this is that for more general cases of $\mudot$, the calculations become too time-consuming, even for computer algebra system. (Probably, some optimization of code developed in \cite{mathemfile} could be helpful in overcoming this obstacle). Also, we have decided to use already known \cite{fedosov-trace},\cite{dobrski-ncgr} formula  for $\tr_F$ to obtain expression for generalized trace, and this formula was derived for $\kappa=0$.

Thus, after all calculations, the following formula can be derived for the generalized Fedosov product of endomorphisms $A,B \in C^\infty (\sympman, \ebund)[[h]]$. (Here $\partial$ stands for covariant derivative which combines $\connsymp$ with $\connbund$, ie{.} acts by means of $\connsymp$ on (co)tangent space and by $\connend$ on endomorphisms. The shortened notation $\omega_{im}\tensor{\curvsymp}{^m_{jkl}}=\curvsymp_{ijkl}$ is also used).
\begin{multline}
\label{ABstar}
A*B=AB-\frac{\rmi}{2}h \Big( \omega^{a b} -4 \rmi g_{(1)}^{a b}\Big)\partial_a A \partial_b B 
+h^2\Bigg(
3 \Big( g_{(1)}^{p} g_{(1)}^{q r} -g_{(2)}^{p q r} \Big)\Big( \partial_{(p} A \partial_q \partial_{r)} B + \partial_{(p} \partial_q A \partial_{r)} B  \Big)
\\
 + \Big(g_{(1)}^{p} g_{(1)}^{q}  - 2 g_{(2)}^{p q} 
+  \big(g_{(1)}^{q r} g_{(1)}^{s a} \omega^{pb}    
+ g_{(1)}^{p r} g_{(1)}^{s a}  \omega^{qb}\big) \Big( \frac{1}{3} \curvsymp_{rsab}   
+ \frac{7}{6}  \curvsymp_{sarb} \Big)  + \frac{\rmi}{2} g_{(1)}^{r s} \curvsymp_{rsab} \omega^{pa} \omega^{qb}  
\\
+ \curvsymp_{sarb} \big( g_{(2)}^{q s a b}  \omega^{pr} + g_{(2)}^{p s a b}  \omega^{qr} \big) + \frac{\rmi}{2} \big(\omega^{qr}+4\rmi g_{(1)}^{q r}\big)  \partial_r g_{(1)}^{p}  - \frac{\rmi}{2} \big( \omega^{pr} - 4 \rmi g_{(1)}^{p r}\big) \partial_r g_{(1)}^{q}\Big) \partial_p A \partial_q B
\\
 + \frac{1}{8} \big( \omega^{ps} -4\rmi g_{(1)}^{ps}\big)\big(\omega^{qr} -4\rmi g_{(1)}^{qr} \big) \partial_p A \partial_q B \curvbund_{rs} 
+ \frac{1}{4} \big( \omega^{ps} -4\rmi g_{(1)}^{ps}\big)\big(\omega^{qr} +4\rmi g_{(1)}^{qr} \big) \partial_p A \curvbund_{rs} \partial_q B 
\\
+ \frac{1}{8} \big( \omega^{ps} +4\rmi g_{(1)}^{ps}\big)\big(\omega^{qr} +4\rmi g_{(1)}^{qr} \big) \curvbund_{rs}  \partial_p A \partial_q B     
\\
+ \frac{\rmi}{2} \partial_s g_{(1)}^{q r} \Big( \big(\omega^{ps} +4\rmi g_{(1)}^{p s} \big)  \partial_r \partial_q A \partial_p B -\big(\omega^{ps} -4\rmi g_{(1)}^{p s} \big)  \partial_p A \partial_r \partial_q B \Big)  
\\
+ \Big( 3 g_{(1)}^{(p s} g_{(1)}^{q r)} - 6 g_{(2)}^{p q r s} - \frac{1}{8}\big(\omega^{ps} -4\rmi g_{(1)}^{ps}\big)\big(\omega^{qr} -4\rmi g_{(1)}^{qr}\big)  \Big) \partial_{(p} \partial_{q)} A \partial_{(r} \partial_{s)} B 
\\
+\Big(2 g_{(1)}^{p q} g_{(1)}^{r s} - 4 g_{(2)}^{p q r s} \Big)  \big(\partial_{(p} \partial_q \partial_r A \partial_{s)} B +\partial_{(p} A \partial_q \partial_r \partial_{s)} B \big)
\Bigg)+O(h^3)
\end{multline}
It is  straightforward to obtain formula for star product of functions from above expression -- one has to put $\curvbund_{ij}\equiv 0$ and interpret $\partial$ as symplectic connection $\connsymp$. Hence, let us focus on $A*X$ and $X*f$, with $X \in C^\infty (\sympman, \bund)[[h]]$, $f \in C^\infty (\sympman, \mathbb{C})[[h]]$. Now, $\partial$ acts by means of $\connbund$ on $X$. It turns out that using (\ref{starproducts}) one gets relations which are formally very similar to that for $A*B$. In fact, the formula for $A*X$ looks like if $B$ was mechanically replaced by $X$ in (\ref{ABstar}), and than all meaningless terms (involving expressions like ``$\partial_p A \partial_q X \curvbund_{rs} $'') were dropped. The analogous statement holds also for $X*f$.  
\begin{multline}
\label{AXstar}
A*X=AX-\frac{\rmi}{2}h \Big( \omega^{a b} -4 \rmi g_{(1)}^{a b}\Big)\partial_a A \partial_b X 
+h^2\Bigg(
3 \Big( g_{(1)}^{p} g_{(1)}^{q r} -g_{(2)}^{p q r} \Big)\Big( \partial_{(p} A \partial_q \partial_{r)} X + \partial_{(p} \partial_q A \partial_{r)} X  \Big)
\\
 + \Big(g_{(1)}^{p} g_{(1)}^{q}  - 2 g_{(2)}^{p q} 
+  \big(g_{(1)}^{q r} g_{(1)}^{s a} \omega^{pb}    
+ g_{(1)}^{p r} g_{(1)}^{s a}  \omega^{qb}\big) \Big( \frac{1}{3} \curvsymp_{rsab}   
+ \frac{7}{6}  \curvsymp_{sarb} \Big)  + \frac{\rmi}{2} g_{(1)}^{r s} \curvsymp_{rsab} \omega^{pa} \omega^{qb}  
\\
+ \curvsymp_{sarb} \big( g_{(2)}^{q s a b}  \omega^{pr} + g_{(2)}^{p s a b}  \omega^{qr} \big) + \frac{\rmi}{2} \big(\omega^{qr}+4\rmi g_{(1)}^{q r}\big)  \partial_r g_{(1)}^{p}  - \frac{\rmi}{2} \big( \omega^{pr} - 4 \rmi g_{(1)}^{p r}\big) \partial_r g_{(1)}^{q}\Big) \partial_p A \partial_q X
\\
+ \frac{1}{4} \big( \omega^{ps} -4\rmi g_{(1)}^{ps}\big)\big(\omega^{qr} +4\rmi g_{(1)}^{qr} \big) \partial_p A \curvbund_{rs} \partial_q X 
+ \frac{1}{8} \big( \omega^{ps} +4\rmi g_{(1)}^{ps}\big)\big(\omega^{qr} +4\rmi g_{(1)}^{qr} \big) \curvbund_{rs}  \partial_p A \partial_q X     
\\
+ \frac{\rmi}{2} \partial_s g_{(1)}^{q r} \Big( \big(\omega^{ps} +4\rmi g_{(1)}^{p s} \big)  \partial_r \partial_q A \partial_p X -\big(\omega^{ps} -4\rmi g_{(1)}^{p s} \big)  \partial_p A \partial_r \partial_q X \Big)  
\\
+ \Big( 3 g_{(1)}^{(p s} g_{(1)}^{q r)} - 6 g_{(2)}^{p q r s} - \frac{1}{8}\big(\omega^{ps} -4\rmi g_{(1)}^{ps}\big)\big(\omega^{qr} -4\rmi g_{(1)}^{qr}\big)  \Big) \partial_{(p} \partial_{q)} A \partial_{(r} \partial_{s)} X 
\\
+\Big(2 g_{(1)}^{p q} g_{(1)}^{r s} - 4 g_{(2)}^{p q r s} \Big)  \big(\partial_{(p} \partial_q \partial_r A \partial_{s)} X +\partial_{(p} A \partial_q \partial_r \partial_{s)} X \big)
\Bigg)+O(h^3)
\end{multline}

\begin{multline}
\label{Xfstar}
X*f=Xf-\frac{\rmi}{2}h \Big( \omega^{a b} -4 \rmi g_{(1)}^{a b}\Big)\partial_a X \partial_b f 
+h^2\Bigg(
3 \Big( g_{(1)}^{p} g_{(1)}^{q r} -g_{(2)}^{p q r} \Big)\Big( \partial_{(p} X \partial_q \partial_{r)} f + \partial_{(p} \partial_q X \partial_{r)} f  \Big)
\\
 + \Big(g_{(1)}^{p} g_{(1)}^{q}  - 2 g_{(2)}^{p q} 
+  \big(g_{(1)}^{q r} g_{(1)}^{s a} \omega^{pb}    
+ g_{(1)}^{p r} g_{(1)}^{s a}  \omega^{qb}\big) \Big( \frac{1}{3} \curvsymp_{rsab}   
+ \frac{7}{6}  \curvsymp_{sarb} \Big)  + \frac{\rmi}{2} g_{(1)}^{r s} \curvsymp_{rsab} \omega^{pa} \omega^{qb}  
\\
+ \curvsymp_{sarb} \big( g_{(2)}^{q s a b}  \omega^{pr} + g_{(2)}^{p s a b}  \omega^{qr} \big) + \frac{\rmi}{2} \big(\omega^{qr}+4\rmi g_{(1)}^{q r}\big)  \partial_r g_{(1)}^{p}  - \frac{\rmi}{2} \big( \omega^{pr} - 4 \rmi g_{(1)}^{p r}\big) \partial_r g_{(1)}^{q}\Big) \partial_p X \partial_q f
\\
+ \frac{1}{8} \big( \omega^{ps} +4\rmi g_{(1)}^{ps}\big)\big(\omega^{qr} +4\rmi g_{(1)}^{qr} \big) \curvbund_{rs}  \partial_p X \partial_q f     
\\
+ \frac{\rmi}{2} \partial_s g_{(1)}^{q r} \Big( \big(\omega^{ps} +4\rmi g_{(1)}^{p s} \big)  \partial_r \partial_q X \partial_p f -\big(\omega^{ps} -4\rmi g_{(1)}^{p s} \big)  \partial_p X \partial_r \partial_q f \Big)  
\\
+ \Big( 3 g_{(1)}^{(p s} g_{(1)}^{q r)} - 6 g_{(2)}^{p q r s} - \frac{1}{8}\big(\omega^{ps} -4\rmi g_{(1)}^{ps}\big)\big(\omega^{qr} -4\rmi g_{(1)}^{qr}\big)  \Big) \partial_{(p} \partial_{q)} X \partial_{(r} \partial_{s)} f 
\\
+\Big(2 g_{(1)}^{p q} g_{(1)}^{r s} - 4 g_{(2)}^{p q r s} \Big)  \big(\partial_{(p} \partial_q \partial_r X \partial_{s)} f +\partial_{(p} X \partial_q \partial_r \partial_{s)} f \big)
\Bigg)+O(h^3)
\end{multline} 
In order to calculate the generalized trace functional, one could use definition (\ref{trdef}). It requires computation of isomorphism $T:\wbund^{D}_\dotmark \to \wbund^{D_F}_\dotmark$ and this is again laborious task, which should be gladly given to some computer algebra system. As in the case of star products we refer interested reader to the \emph{Mathematica} file \cite{mathemfile} for details. Here only the final result is presented. 
\begin{multline}
\label{trgenres}
\tr Q_\ebund (A) = \int_\sympman \Tr \Bigg( 
A + h \left(\frac{\rmi}{2} \curvbund_{ab} \omega^{ab} -  \partial_a g_{(1)}^{a} +  \partial_b \partial_a g_{(1)}^{a b}\right) A 
+ h^2 \Bigg( 
-  \frac{3}{8}   \curvbund_{ab} \curvbund_{cd} \omega^{[ab} \omega^{cd]} 
\\
 - \frac{\rmi}{2}  \omega^{bc}  \partial_a \left(g_{(1)}^{a} \curvbund_{bc} \right) - \rmi  \omega^{bc} \partial_c \left(g_{(1)}^{a} \curvbund_{ab} \right)  + \frac{\rmi}{2} \omega^{cd}  \partial_b\partial_a \left(g_{(1)}^{a b} \curvbund_{cd} \right) + \frac{\rmi}{4} \omega^{cd}  \partial_b \partial_d \left(g_{(1)}^{a b} \curvbund_{ac} \right) 
\\
- \rmi \omega^{cd}  \partial_d \left(g_{(1)}^{a b}  \partial_b \curvbund_{ac}\right) + \frac{3\rmi}{4} \omega^{cd} \partial_d \partial_b \left(g_{(1)}^{a b} \curvbund_{ac} \right) +  \frac{1}{48}   \omega^{ab} \omega^{cd} \omega^{ep} \partial_d \partial_b \curvsymp_{acep}
\\
+ \frac{1}{64}   \tensor{\curvsymp}{^{d}_{aeq}} \tensor{\curvsymp}{^{a}_{dpr}}  \omega^{[eq} \omega^{pr]}
 +   \partial_b \partial_p \left(g_{(2)}^{c d e p} \curvsymp_{cdae} \omega^{ab} + g_{(1)}^{a b} g_{(1)}^{c d}\left( \frac{3}{2} \curvsymp_{acde} + \frac{1}{2} \curvsymp_{cead} \right)\omega^{ep}\right)  
\\
- \omega^{de} \partial_e \left(g_{(1)}^{a} g_{(1)}^{b c} \left( \curvsymp_{abcd} + \curvsymp_{bcad} 	\right)  + g_{(2)}^{ a b c} \curvsymp_{cadb} \right) 
- \omega^{ep} \partial_p \left(g_{(2)}^{c b a d}  \partial_d \curvsymp_{cbea}+g_{(1)}^{a b} g_{(1)}^{c d}  \partial_d \curvsymp_{abce}\right)
\\
+ \omega^{cd} \partial_d \left(g_{(1)}^{a b}  \left(\frac{2}{3} \curvsymp_{cepb}   \partial_a g_{(1)}^{e p} +\frac{1}{2} \curvsymp_{aebp}  \partial_c g_{(1)}^{e p}\right) \right)    
 -   \partial_a g_{(2)}^{a}  -   \partial_b \left(g_{(1)}^{a} \partial_a g_{(1)}^{b}\right) 
\\
- \partial_a \left(g_{(1)}^{b c} \partial_c \partial_b g_{(1)}^{a}\right) 
  +   \partial_a \partial_c \left(g_{(1)}^{b} \partial_b g_{(1)}^{a c}+2 g_{(1)}^{b c} \partial_b g_{(1)}^{a}\right)
  - \frac{1}{3}    \partial_b \partial_d \partial_c \left(2 g_{(1)}^{a b} \partial_a g_{(1)}^{c d} + 4g_{(1)}^{a c} \partial_a g_{(1)}^{b d}\right)   
\\
 +   \partial_d \partial_c \left(g_{(1)}^{a b} \partial_b \partial_a g_{(1)}^{c d}\right) 
+ \partial_a \partial_b g_{(2)}^{a b} - \partial_a \partial_b \partial_c g_{(2)}^{a b c}
 +   \partial_a \partial_b \partial_c \partial_d g_{(2)}^{a b c d} 
\Bigg)A +O(h^3)
\Bigg)\frac{\omega^n}{n!}
\end{multline}
As in the case of star product, the trace of function can be obtained from right hand side of above formula, by dropping terms involving curvature $\curvbund_{ij}$, and replacing $A$ by function.

\subsection{Example -- symmetric part of noncommutativity tensor}
From the general structure of definition (\ref{gdef}) of $g$ it follows that our setting covers the case of ``Moyal product with nonsymmetric $\omega^{ij}$''. More precisely, if one demands that $\gcirc$ is given by (\ref{moyal}) with $\omega^{ij}$ replaced by $m^{ij}=\omega^{ij}+g^{ij}$ for antisymmetric (and with symplectic inverse) $\omega^{ij}$ and symmetric $g^{ij}$, then $g$ can be taken as
\begin{equation}
g=\exp \left(\frac{\rmi h}{4} g^{ij} \frac{\partial}{\partial y^i}\frac{\partial}{\partial y^j}\right)
\end{equation}
This means that, up to $h^2$, the only nonzero coefficients $g_{(s)}^{i_1 \dots i_k}$ are
\begin{equation}
\label{symsubst}
g_{(1)}^{i j}= \frac{\rmi}{4} g^{ij} \quad \text{and} \quad g_{(2)}^{i j k l}= -\frac{1}{32} g^{(ij} g^{kl)}
\end{equation}
Two cases can be considered.
\subsubsection{Case of arbitrary symplectic connection}
If there are no further assumptions about $\connsymp$, then after substitution (\ref{symsubst}), and due to symmetries of $\curvsymp_{ijkl}$, the formula (\ref{ABstar}) can be brought to the following form.
\begin{multline}
A*B=AB-\frac{\rmi}{2}h \Big( \omega^{a b} + g^{a b}\Big)\partial_a A \partial_b B 
+h^2\Bigg(
-\frac{1}{8} g^{r s} \curvsymp_{r s a b} \Big(
\omega^{a (p} g^{q) b} + \omega^{pa} \omega^{qb}  
 \Big) \partial_p A \partial_q B
\\
 + \frac{1}{8} \big( \omega^{ps} + g^{ps}\big)\big(\omega^{qr} + g^{qr} \big) \partial_p A \partial_q B \curvbund_{rs} 
+ \frac{1}{4} \big( \omega^{ps} + g^{ps}\big)\big(\omega^{qr} - g^{qr} \big) \partial_p A \curvbund_{rs} \partial_q B 
\\
+ \frac{1}{8} \big( \omega^{ps} - g^{ps}\big)\big(\omega^{qr} - g^{qr} \big) \curvbund_{rs}  \partial_p A \partial_q B     
- \frac{1}{8} \partial_s g^{q r} \Big( \big(\omega^{ps} - g^{p s} \big)  \partial_r \partial_q A \partial_p B -\big(\omega^{ps} + g^{p s} \big)  \partial_p A \partial_r \partial_q B \Big) 
\\ 
- \frac{1}{8}\big(\omega^{ps} + g^{ps}\big)\big(\omega^{qr} + g^{qr}\big)   \partial_{(p} \partial_{q)} A \partial_{(r} \partial_{s)} B 
\Bigg)+O(h^3)
\end{multline}
The trace functional reads now
\begin{multline}
\tr Q_\ebund (A) = \int_\sympman \Tr \Bigg( 
A +  \frac{\rmi h}{2}\left( \curvbund_{ab} \omega^{ab} +  \frac{1}{2}\partial_b \partial_a g^{a b}\right) A 
+ h^2 \Bigg( 
-  \frac{3}{8}   \curvbund_{ab} \curvbund_{cd} \omega^{[ab} \omega^{cd]} 
\\
  - \frac{1}{8} \omega^{cd}  \partial_b\partial_a \left(g^{a b} \curvbund_{cd} \right) - \frac{1}{16} \omega^{cd}  \partial_b \partial_d \left(g^{a b} \curvbund_{ac} \right) 
+\frac{1}{4} \omega^{cd}  \partial_d \left(g^{a b}  \partial_b \curvbund_{ac}\right) 
\\
- \frac{3}{16} \omega^{cd} \partial_d \partial_b \left(g^{a b} \curvbund_{ac} \right) 
+  \frac{1}{48}   \omega^{ab} \omega^{cd} \omega^{ep} \partial_d \partial_b \curvsymp_{acep}
+ \frac{1}{64}   \tensor{\curvsymp}{^{d}_{aeq}} \tensor{\curvsymp}{^{a}_{dpr}}  \omega^{[eq} \omega^{pr]}
\\
-\frac{1}{32}  \partial_b \partial_p \left( g^{(c d}g^{ e p)} \curvsymp_{cdae} \omega^{ab} + g^{a b} g^{c d}\left( 3 \curvsymp_{acde} +  \curvsymp_{cead} \right)\omega^{ep}\right)  
\\
+\frac{1}{32} \omega^{ep} \partial_p \left(g^{(c b}g^{ a d)}  \partial_d \curvsymp_{cbea}+2 g^{a b} g^{c d}  \partial_d \curvsymp_{abce}\right)
\\
-\frac{1}{16} \omega^{cd} \partial_d \left(g^{a b}  \left(\frac{2}{3} \curvsymp_{cepb}   \partial_a g^{e p} +\frac{1}{2} \curvsymp_{aebp}  \partial_c g^{e p}\right) \right)    
  + \frac{1}{48}    \partial_b \partial_d \partial_c \left(2 g^{a b} \partial_a g^{c d} + 4 g^{a c} \partial_a g^{b d}\right)   
  \\
 -\frac{1}{16}   \partial_d \partial_c \left(g^{a b} \partial_b \partial_a g^{c d}\right) 
 -\frac{1}{32}   \partial_a \partial_b \partial_c \partial_d \left( g^{(a b}g^{ c d)} \right) 
\Bigg)A +O(h^3)
\Bigg)\frac{\omega^n}{n!}
\end{multline}
\subsubsection{Case of symplectic connection preserving $g^{ij}$}
It can be furthermore demanded that symplectic connection preserves $g^{ij}$, ie. $\connsymp_k g^{ij}=0$.  Notice, that in this specific way we arrive nearly to generalizations considered in \cite{bering} and also to Wick-type setting of \cite{dolglyaksharap}. It is assumed there, that the fiberwise product is given by Moyal formula with nonsymmetric noncommutativity tensor $m^{ij}=\omega^{ij}+g^{ij}$, where $\omega^{ij}$ is nondegenerate, and that torsionfree connection in the tangent bundle preserves $m^{ij}$. This implies, that the connection preserves separately $\omega^{ij}$ and $g^{ij}$, and that inverse of $\omega^{ij}$ is symplectic. However, both \cite{bering} and \cite{dolglyaksharap} use definition of $\gbrc$ different from (\ref{gconnsqr}), setting it to $\gbrc=\brc=\frac{1}{4}\omega_{im}\tensor{\curvsymp}{^m_{jkl}}y^i y^j \rmd x^k \wedge \rmd x^l$. On the other hand, here $\gbrc=g^{-1}\brc$ which means that additional term, proportional to $h g^{ij} \omega_{im}\tensor{\curvsymp}{^m_{jkl}} dx^k \wedge dx^l$, appears. As this term is scalar, it does not change commutators. But when used \emph{outside} commutators, and this is the case of theorem \ref{fedo_abel}, it produces some input. Analyzing expressions (\ref{omeq1}) and (\ref{omeq2}) one can observe, that to restore full compatibility with \cite{bering} and \cite{dolglyaksharap} nonzero correction $\kappa$ proportional to $g^{ij} \omega_{im}\tensor{\curvsymp}{^m_{jkl}} dx^k \wedge dx^l$ must be introduced, and such a case is outside of the scope of this section. Notice also, that if $2$-form $g^{ij} \omega_{im}\tensor{\curvsymp}{^m_{jkl}} dx^k \wedge dx^l$ is cohomologically nontrivial, than the aforementioned correction $\kappa$ changes equivalence class of the star product. This observation is directly related to results on equivalence of Wick- and Weyl-type star products derived in \cite{dolglyaksharap}.

For the present setting one is able to simplify a bit the formula for the star product of endomorphisms
\begin{multline}
A*B=AB-\frac{\rmi}{2}h \Big( \omega^{a b} + g^{a b}\Big)\partial_a A \partial_b B 
+h^2\Bigg(
-\frac{1}{8} g^{r s} \curvsymp_{r s a b} \Big(
\omega^{a (p} g^{q) b} + \omega^{pa} \omega^{qb}  
 \Big) \partial_p A \partial_q B
\\
 + \frac{1}{8} \big( \omega^{ps} + g^{ps}\big)\big(\omega^{qr} + g^{qr} \big) \partial_p A \partial_q B \curvbund_{rs} 
+ \frac{1}{4} \big( \omega^{ps} + g^{ps}\big)\big(\omega^{qr} - g^{qr} \big) \partial_p A \curvbund_{rs} \partial_q B 
\\
+ \frac{1}{8} \big( \omega^{ps} - g^{ps}\big)\big(\omega^{qr} - g^{qr} \big) \curvbund_{rs}  \partial_p A \partial_q B     
- \frac{1}{8}\big(\omega^{ps} + g^{ps}\big)\big(\omega^{qr} + g^{qr}\big)   \partial_{(p} \partial_{q)} A \partial_{(r} \partial_{s)} B 
\Bigg)+O(h^3)
\end{multline}
and to bring the trace to somehow more friendly form 
\begin{multline}
\tr Q_\ebund (A) = \int_\sympman \Tr \Bigg( 
A +  \frac{\rmi h}{2} \curvbund_{ab} \omega^{ab} A 
+ h^2 \Bigg( 
-  \frac{3}{8}   \curvbund_{ab} \curvbund_{cd} \omega^{[ab} \omega^{cd]} 
\\
-   g^{a b} \omega^{cd} \left( \frac{1}{8} \partial_a \partial_b \curvbund_{cd}  
+ \frac{1}{16}   \partial_b \partial_d \curvbund_{ac}  
- \frac{1}{16}   \partial_d \partial_b \curvbund_{ac} \right) 
\\
+  \frac{1}{48}   \omega^{ab} \omega^{cd} \omega^{ep} \partial_d \partial_b \curvsymp_{acep}
+ \frac{1}{64}   \tensor{\curvsymp}{^{d}_{aeq}} \tensor{\curvsymp}{^{a}_{dpr}}  \omega^{[eq} \omega^{pr]}
\\
+ g^{a b} g^{c d} \omega^{e p} \left(-\frac{3}{32} \partial_d \partial_p \curvsymp_{a c b e} + \frac{1}{32} \partial_d \partial_p \curvsymp_{a e b c} + \frac{1}{16} \partial_p \partial_d \curvsymp_{a b c e}\right) 
\Bigg)A 
+O(h^3)
\Bigg)\frac{\omega^n}{n!}
\end{multline}
\section{Final remarks}
The general scheme for our approach (as well as for \cite{tosiekacta3}) was to establish isomorphism $g$ to standard fiberwise Moyal product, and then to pull-back \emph{almost} all structures required for Fedosov construction. One can immediately observe, that if we pull-back \emph{exactly} everything, then the result is strictly trivial -- the star products are precisely those of generic Fedosov quantization. The key point is that we do not modify mappings $\delta$ and $\delta^{-1}$ leaving ``Poincare decomposition'' untouched. However, the $\delta$ operator should remain $+1$-derivation for the new fiberwise product $\gcirc$. This can be achieved by imposing  condition 
$g \delta = \delta g $
, and definition (\ref{gdef}) is consistent with it. (The same condition causes the relation $\delta \gconn^\wbdot + \gconn^\wbdot \delta=0$ to hold true, which appears to be important in constructing Abelian connection).

The deformation quantization variant elaborated in this paper is isomorphic to the generic Fedosov quantization of \cite{fedosovbook}, as can be easily seen both by the very construction, and due to general results on deformation quantization on symplectic manifolds \cite{xu}. This however does not mean that presented results are trivial, especially from the point of view of possible applications in physics. Indeed, one can observe that introduced degrees of freedom enter rather nontrivially into formulas (\ref{ABstar}), (\ref{AXstar}), (\ref{Xfstar}) and (\ref{trgenres}). Hence, when employed into any model-building, they should result in some physical input (compare also with \cite{cahenflatogutt}). For instance, it could be quite interesting to study further consequences of obtained results in operator ordering problems, as in \cite{tosiekacta3}. The most straightforward (but not the only one) field-theoretic application can be conjectured from our example, with $g^{ij}$ given the meaning of a metric on the spacetime. The formalism derived here can be rather safely used for modeling noncommutativity of the spacetime, and such kind of considerations are planned to be covered in author's next paper. Notice however, that when passing to infinite-dimensional case of quantization of fields, our approach should be taken cautiously.
It is known that in such situation equivalence between different quantizations can be broken due to divergences, and this could imply that the mapping (\ref{gdef}) would become ill-defined (compare with \cite{dolglyaksharap} for more detailed comments on such issue).
\section{Acknowledgments}
I am grateful to professor Maciej Przanowski for his interest in the present paper.
This work was supported by the Polish Ministry of Science and Higher Education grant no.~IP20120220072 within Iuventus Plus programme.

\end{document}